\documentclass[11pt]
{article}
\usepackage[margin=1in]{geometry}
\usepackage[utf8]{inputenc}
\usepackage{amssymb}
\usepackage{bbm}	
\usepackage{amsthm,amssymb,amsmath}
\usepackage{mathtools}
\usepackage{thmtools,thm-restate}
\usepackage{tikz}
\usetikzlibrary{arrows,automata}
\usepackage{multirow}
\usepackage{tabularx}
\usepackage[ruled,lined,boxed]{algorithm2e}
\usepackage{algorithmic}
\usepackage{comment}
\usepackage[caption=false]{subfig}
\usepackage{booktabs}
\usepackage{libertine}
\usepackage{authblk}

\usepackage{thm-restate}
\usepackage[libertine]{newtxmath} % must load after ams packages
\usepackage[scaled=0.96]{zi4} % use Inconsolata as the typewriter font

\usepackage{ dsfont, microtype, xcolor, paralist}

\usepackage[ocgcolorlinks]{hyperref} % Option ocgcolorlinks makes links only colored on screen and not on printouts
\colorlet{DarkRed}{red!50!black}
\colorlet{DarkGreen}{green!50!black}
\colorlet{DarkBlue}{blue!50!black}
\hypersetup{
	linkcolor = DarkRed,
	citecolor = DarkGreen,
	urlcolor = DarkBlue,
	bookmarksnumbered = true,
	linktocpage = true
}

\newcommand{\HG}[1]{\textcolor{red}{[HG:#1]}}

\usepackage{xspace}
\newcommand{\PName}{\textbf{DRSO}\xspace}
\let\epsilon\varepsilon
\let\eps\varepsilon

\usepackage{xspace}

\usepackage{tikz}
\usetikzlibrary{decorations.pathreplacing}
\usetikzlibrary{plotmarks}
\usetikzlibrary{positioning,automata,arrows}
\usetikzlibrary{shapes.geometric}
\usetikzlibrary{decorations.markings}
\usetikzlibrary{positioning, shapes, arrows}

\PassOptionsToPackage{usenames,dvipsnames,svgnames}{xcolor}
\definecolor{orange}{RGB}{235,90,0}
\definecolor{darkorange}{RGB}{175,30,0}
\definecolor{turkis}{RGB}{131,182,182}
\definecolor{darkturkis}{RGB}{31,82,82}
\definecolor{green}{RGB}{102,180,0}
\definecolor{darkgreen}{RGB}{51,90,0}
\definecolor{myblue}{RGB}{0,0,213}
\definecolor{mydarkblue}{RGB}{0,0,100}
\definecolor{mybrightblue}{HTML}{74B0E4}
\definecolor{mybrighterblue}{HTML}{B3EAFA}
\definecolor{lila}{RGB}{102,0,102}
\definecolor{darkred}{RGB}{139,0,0}
\definecolor{darkyellow}{RGB}{188,135,2}
\definecolor{brightgray}{RGB}{200,200,200}
\definecolor{darkgray}{RGB}{50,50,50}
\definecolor{amaranth}{rgb}{0.9, 0.17, 0.31}
\definecolor{alizarin}{rgb}{0.82, 0.1, 0.26}
\definecolor{amber}{rgb}{1.0, 0.75, 0.0}
\definecolor{green(ryb)}{rgb}{0.4, 0.69, 0.2}
\definecolor{hanblue}{rgb}{0.27, 0.42, 0.81}
\definecolor{grannysmithapple}{rgb}{0.66, 0.89, 0.63}

\usepackage{amsthm}

\newtheorem{theorem}{Theorem}[section]
\newtheorem{lemma}{Lemma}[section]

\newtheorem{corollary}{Corollary}[section]
\newtheorem{observation}{Observation}[section]

\newtheorem{claim}{Claim}[section]

\newtheorem{assumption}{Assumption}[section]
\newtheorem{lemma-rstbl}{Lemma}[section]
\newtheorem{obs-rstbl}{Observation}[section]
\newtheorem{theorem-rstbl}{Theorem}[section]

\usepackage{environ}
\NewEnviron{cproblem}[1]{%
\begin{center}\fbox{\parbox{5.5in}{%
    {\centering\scshape #1\par}%
    \parskip=1ex
    \everypar{\hangindent=1em}%
    \BODY
}}\end{center}}

\title{Budgeted Out-tree Maximization with Submodular Prizes}
\author[1]{Gianlorenzo D'Angelo}
\author[1]{Esmaeil Delfaraz}
\author[2]{Hugo Gilbert}
\affil[1]{\normalsize Gran Sasso Science Institute, L'Aquila, Italy}
\affil[2]{Université Paris-Dauphine, Université PSL, CNRS, LAMSADE, 75016 Paris, France}
\date{}

\begin{document}

\maketitle

\begin{abstract}
%We consider a variant of the prize collecting Steiner tree problem on \emph{directed graphs} in which each node is associated with a cost, each set of nodes is associated with a monotone submodular prize function, and we are given a specific node $r$, called root, and a budget $B$.  The aim is to find a out-subtree rooted at $r$ which costs at most $B$ and maximizes the prize function.
%
We consider a variant of the prize collecting Steiner tree problem in which we are given a \emph{directed graph} $D=(V,A)$, a monotone submodular prize function $p:2^V \rightarrow \mathbb{R}^+ \cup \{0\}$, a cost function $c:V \rightarrow \mathbb{Z}^{+}$, a root vertex $r \in V$, and a budget $B$. The aim is to find an out-subtree $T$ of $D$ rooted at $r$ that costs at most $B$ and maximizes the prize function. We call this problem \emph{Directed Rooted Submodular Tree} (\PName).

For the case of undirected graphs and additive prize functions, Moss and Rabani [SIAM\ J.\ Comput.\ 2007] gave an algorithm that guarantees an $O(\log |V|)$-approximation factor if a violation by a factor 2 of the budget constraint is allowed. Bateni et al.\ [SIAM J.\ Comput.\ 2018] improved the budget violation factor to $1+\epsilon$ at the cost of an additional approximation factor of $O(1/\epsilon^2)$, for any $\epsilon\in (0,1]$.
For directed graphs, Ghuge and Nagarajan [SODA\ 2020] gave an optimal quasi-polynomial-time $O\left(\frac{\log n'}{\log \log n'}\right)$-approximation algorithm, where $n'$ is the number of vertices in an optimal solution, for the case in which the costs are associated to the edges.

In this paper, we give a polynomial-time algorithm for \PName that guarantees an approximation factor of $O(\sqrt{B}/\epsilon^3)$ at the cost of a budget violation of a factor $1+\epsilon$, for any $\epsilon \in (0,1]$. The same result holds for the edge-cost case, to the best of our knowledge this is the first polynomial-time approximation algorithm for this case.
We further show that the unrooted version of \PName can be approximated to a factor of $O(\sqrt{B})$ without budget violation, which is an improvement over the factor $O(\Delta \sqrt{B})$ given in~[Kuo et al.\ IEEE/ACM\ Trans.\ Netw.\ 2015] for the undirected and unrooted case, where $\Delta$ is the maximum degree of the graph.
%
%Finally, we show how to use our algorithm to improve the approximation bounds of several related problems, including the undirected and quota versions of \PName, the maximum budgeted connected set cover problem, and the budgeted sensor cover problem.
Finally, we provide some new/improved approximation bounds for several related problems, including the additive-prize version of \PName, the maximum budgeted connected set cover problem, and the budgeted sensor cover problem.
\end{abstract}

\section{Introduction}
\textit{Prize collecting Steiner tree problems} (\textbf{PCSTP}) have been extensively studied due to their applications in designing computer and telecommunication networks, VLSI design, computational geometry, wireless mesh networks, and cancer genome studies~\cite{cheng2004steiner, gao2018algorithm,hochbaum2020approximation, kuo2014maximizing,vandin2011algorithms}.
%Another interesting application of \textbf{PCSTP} is  related to cancer genome studies~\cite{bayat2020multi, hochbaum2020approximation, sun2017node, vandin2011algorithms}.
%. In particular, the interaction network (proteinDNA interactions) is modeled by a graph in which vertices represent individual proteins (and their associated genes) and %edges represent (pairwise) protein-protein or proteinDNA interactions.
%the goal is to find a connected subgraph subject to some desirable objective functions~\cite{bayat2020multi, hochbaum2020approximation, sun2017node, vandin2011algorithms}, e.g., finding $k$-connected genes that is mutated in the largest number of samples (patients).
%
Very interesting polynomial-time constant/poly-logarithmic approximation algorithms have been proposed for many variants of \textbf{PCSTP} when the graph is undirected~\cite{archer2011improved, bateni2018improved, garg2005saving, goemans1995general, johnson2000prize, konemann2013lmp, paul2020budgeted}. However, these problems are usually much harder on directed graphs. For instance, there is a simple polynomial-time $2$-approximation algorithm for the \textit{undirected Steiner tree} problem, but no quasi-polynomial-time algorithm for the \textit{directed Steiner tree} problem achieving an approximation ratio of $o\left(\frac{\log^2 k}{\log \log k}\right)$ exists, unless $NP \subseteq \bigcap_{0 <\epsilon< 1} \text{ZPTIME}(2^{n^\epsilon})$ or the Projection Game Conjecture is false~\cite{grandoni2019log2}, where $k$ is the number of terminal nodes.% (\ED{I think we have the same situation for the edge-cost version of our problem})

Some of the most relevant variants of \textbf{PCSTP} are represented by prize collecting problems with budget constraints. In such problems, we are usually given a graph with prizes and costs on the nodes and the goal is to find a tree that maximizes the sum of the prize of its nodes, while keeping the total cost bounded by a given budget. %~\cite{bateni2018improved, guha1999efficient, kortsarz2009approximating, moss2007approximation}. 
Guha et al.~\cite{guha1999efficient} introduced the case in which the graph is undirected and the goal is to find a tree that  contains a distinguished vertex, called root, respects the budget constraint, and maximizes the prize, we call this problem \emph{Undirected Rooted Additive Tree} (\textbf{URAT}). They gave an algorithm that achieves an $O(\log^2 n)$-approximation factor, where $n$ is the number of nodes in the graph, but the computed solution requires a factor-2 violation of the budget constraint. Moss and Rabani~\cite{moss2007approximation} and Bateni et al.~\cite{bateni2018improved} further investigated \textbf{URAT} and improved the results from the approximability point of view. The former paper improved the approximation factor to $O(\log n)$, with the same budget violation, and the latter one improved the budget violation factor to $1+\epsilon$ to obtain an approximation factor of $O\left(\frac{1}{\epsilon^2}\log n\right)$, for any $\epsilon\in (0,1]$.
Kortsarz and Nutov~\cite{kortsarz2009approximating} showed that the unrooted version of \textbf{URAT}, so does \textbf{URAT}, admits no $o(\log \log n)$-approximation algorithm, unless $NP \subseteq DTIME(n ^{\text{polylog}(n)})$, even if the algorithm is allowed to violate the budget constraint by a factor equal to a universal constant.  

In this paper, we consider a generalization  of \textbf{URAT} on directed graphs. We are given a directed graph, where each node is associated with a cost, and the prize is defined by a monotone submodular function on the subsets of nodes, and the goal is to find an \emph{out-tree} (a.k.a. \emph{out-arborescence}) rooted at a specific vertex $r$ with the maximum prize such that the total cost of all vertices in the out-tree is no more than a given budget. We term this problem \emph{Directed Rooted Submodular Tree} (\PName). A closely related problem, called \emph{Submodular Tree Orienteering} (\textbf{STO}), has been recently introduced by Ghuge and Nagarajan~\cite{ghuge2020quasi}. \textbf{STO} is the same problem as \PName except that edges and not nodes have costs. They provided a tight quasi-polynomial-time $O(\frac{\log n'}{\log \log n'})$-approximation algorithm that requires $(n \log B)^{O(\log^{1+\epsilon} n')}$ time, where $n'$ is the number of vertices in an optimal solution and $B$ is the budget constraint.

\subparagraph*{Contribution.} By extending some ideas of Kuo et al.~\cite{kuo2014maximizing} and Bateni et al.~\cite{bateni2018improved}, we design a polynomial-time $O(\sqrt{B}/\epsilon^{3})$-approximation algorithm for \PName, violating the budget constraint $B$ by a factor of at most $1+\epsilon$, for any $\epsilon \in (0, 1]$ (Section~\ref{sec:apxalgo}). Our technique can be used to obtain the same result for \textbf{STO} (Section~\ref{sec:STO}). To our knowledge, this is the first polynomial-time approximation algorithm for \textbf{STO}. We also show that, for any $1+\epsilon$ budget violation, with $\epsilon \in (0, 1]$, our approach provides an $O(\sqrt{B}/\epsilon^{2})$-approximation algorithm for the special cases of \PName and \textbf{STO} where the prize function is additive (Section~\ref{sec:variants}). We also consider the unrooted version of \PName and give an $O(\sqrt{B})$-approximation algorithm without budget violation (Section~\ref{sec:unrooted}), which is an improvement over the factor $O(\Delta \sqrt{B})$~\cite{kuo2014maximizing} for the undirected and unrooted version of \PName, where $\Delta$ is the maximum node-degree.

%By using some materials from the literature,

Finally, we study some variants of \PName on undirected graphs. We show that, for any $1+\epsilon$ budget violation, \textbf{URAT} admits an $O(\Delta/\epsilon^{2})$-approximation algorithm, while its quota version admits a $2\Delta$-approximation algorithm. Next, we present some approximation results for some variants of the connected maximum coverage problem, which improve over the bounds given by Ran et al.~\cite{ran2016approximation}. Finally, we provide two approximation algorithms for the Budgeted Sensor Cover problem, which result in an improvement
%over the factor $8(\lceil 2\sqrt{2} C\rceil+1)^2/(1-1/e)$, where $C=O(1)$~\cite{huang2020approximation}.
to the literature~\cite{kuo2014maximizing, ran2016approximation, xu2021throughput, yu2019connectivity}.
We discuss these results in Section~\ref{sec:variants}.
%All missing proofs can be found in the Appendix.
%
%over the factors $O(\log{n})$~\cite{gao2018algorithm, ran2016approximation, huang2020approximation}, $O(\sqrt{B})$~\cite{gao2018algorithm, kuo2014maximizing, xu2021throughput} and $8(\lceil 2\sqrt{2} C\rceil+1)^2/(1-1/e)$~\cite{huang2020approximation}, where $C=O(1)$ is an input parameter. All missing proofs can be found in the Appendix.
\subparagraph*{Related Work.}%\label{secRW}
Many variants of Prize collecting Steiner Tree problems have been investigated. Here we list those that are more closely related to our study. Further related work is reported in the Appendix.%  (some of them are deferred to the Appendix).

%Recently, Ghuge and Nagarajan~\cite{ghuge2020quasi} investigated Submodular Tree Orienteering (\textbf{STO}), that is the same problem as \PName except that edges and not nodes have costs. They provided a quasi-polynomial time $O(\log n'/\log \log n')$-approximation algorithm for \textbf{STO} that runs in $(n \log B)^{O(\log^{1+\epsilon} n')}$ time for any constant $\epsilon > 0$, where $n'$ is the number of vertices in an optimal solution. The authors mentioned that this factor is tight for \textbf{STO} in quasi-polynomial time. D'Angelo et al.~\cite{dangelo2022computation} studied a special case of \textbf{STO} in Liquid Democracy systems.%, where the costs on edges are either $0$ or $1$ and the prize function is additive. The authors provided a polynomial-time bicriteria $(1+\epsilon, \frac{8n}{\epsilon^2 B})$-approximation algorithm for this problem, which can be extended to the special case of \textbf{STO} in which the prize function is additive.

%\gd{TODO: move this sentence later.}
%D'Angelo et al.~\cite{dangelo2022computation} studied a special case of \textbf{STO} in Liquid Democracy systems.

%Many variants of \textbf{STO} on undirected graphs have been investigated.
Kuo et al.~\cite{kuo2014maximizing} studied the unrooted version of \PName on undirected graphs called \emph{Maximum Connected Submodular function with Budget constraint} (\textbf{MCSB}). They provided an $O(\Delta \sqrt{B})$-approximation algorithm for \textbf{MCSB}, where $\Delta$ is the maximum degree of the graph. Vandin et al.~\cite{vandin2011algorithms} provided a $(\frac{2e-1}{e-1}r)$-approximation algorithm for a special case of the same problem, where $r$ is the radius of an optimal solution. This problem coincides with the connected maximum coverage problem in which each set has cost one. Ran et al.~\cite{ran2016approximation} presented an $O(\Delta\log{n})$-approximation algorithm for a special case of the connected maximum coverage problem. Hochbaum and Rao~\cite{hochbaum2020approximation} investigated \textbf{MCSB} in which each vertex costs $1$ and provided an approximation algorithm with factor $\min\{1/((1-1/e)(1/R-1/B)), B\}$, where $R$ is the radius of the graph. Chen et al.~\cite{chen2020optimal} investigated the edge-cost version of \textbf{MCSB}. One of the applications of \textbf{MCSB} is a problem in wireless sensor networks called \emph{Budgeted Sensor Cover Problem} (\textbf{BSCP}), where the goal is to find a set of $B$ connected sensors to maximize the number of covered users, for a given $B$. Kuo et al.~\cite{kuo2014maximizing} provided a $5(\sqrt{B}+1)/(1-1/e)$-approximation algorithm for \textbf{BSCP}, which was improved by Xu et al.~\cite{xu2021throughput} to $\lfloor\sqrt{B}\rfloor/(1-1/e)$. Huang et al.~\cite{huang2020approximation} proposed a $8(\lceil 2\sqrt{2} C\rceil+1)^2/(1-1/e)$-approximation algorithm for \textbf{BSCP}, where $C=O(1)$. %Gao et al.~\cite{gao2018algorithm} presented two approximation algorithms with factors $O(\log{n})$ and $O(\sqrt{B})$ for a generalization of \textbf{BSCP} in which each user requires a different .
%
%closely related problem.%, $R_c\le R_s$, $R_s$ and $R_c$ are the sensing range and communication range of sensors, respectively. %As \textbf{STO} is a more general case of \textbf{MCSB}, the technique by Ghuge and Nagarajan~\cite{ghuge2020quasi} provides a quasi-polynomial time $O(\frac{\log n'}{\log \log n'})$-approximation algorithm for all of the above problems.% and their directed versions.

%Many variants of \textbf{MCSB} when the prize function is additive have been investigated. Guha et al.~\cite{guha1999efficient} introduced \textbf{NW-RBP} which is the same as \textbf{MCSB} except that the solution must contain a distinguished vertex $r$ and the prize function is additive. They proposed a bicriteria $(2, O(\log^2 {n}))$-approximation algorithm for \textbf{NW-RBP}. Bateni et al.~\cite{bateni2018improved} provided a bicriteria $O(1+\epsilon, O(\log(n)/\epsilon^2))$-approximation algorithm for \textbf{NW-RBP} and an $O(\log{n})$-approximation algorithm for its unrooted version.
%Many variants of \PName when the prize function is additive and the graph is undirected have been investigated. 
Johnson et al.~\cite{johnson2000prize} introduced an edge-cost variant of \PName on undirected graphs, where the prize function is additive, called \textbf{E-URAT}. They showed that there exists a $(5+\epsilon)$-approximation algorithm for the unrooted version of \textbf{E-URAT} using Garg's $3$-approximation algorithm~\cite{garg3} for the $k$-MST problem, and observed that a $2$-approximation for $k$-MST would lead to a 3-approximation for \textbf{E-URAT}. 
%The approximation bound for the unrooted version of \textbf{E-URAT} was later improved to $4+\epsilon$ by Levin~\cite{levin2004better}.
This observation along with the Garg's $2$-approximation algorithm~\cite{garg2005saving} for $k$-MST yield a $3$-approximation algorithm for the unrooted version of \textbf{E-URAT}.  %So, Garg's $2$-approximation algorithm~\cite{garg2005saving} for $k$-MST provides such a guarantee.
Recently, Paul et al.~\cite{paul2020budgeted} provided a polynomial-time $2$-approximation algorithm for \textbf{E-URAT}.

\begin{comment}
\begin{table}[ht!]
\begin{center}
\begin{tabular}{ |c|c| } 
 \hline
 Problem & Best Bound \\ \hline
 \textbf{STO} & $O(\frac{\log{n}}{\log{\log{n}}})$~\cite{ghuge2020quasi} (tight) \\ \hline
 \textbf{DTO} & $O(\frac{\log{n}}{\log{\log{n}}})$~\cite{ghuge2020quasi} (tight) \\ \hline
 \textbf{DSTP} & $O(\frac{\log^2{k}}{\log{\log{k}}})$~\cite{ghuge2020quasi, grandoni2019log2} (tight) \\ \hline
 \textbf{NW-PCST} & $O(\log{n})$~\cite{bateni2018improved, konemann2013lmp} (tight) \\ \hline
  \textbf{EW-PCST} & $2-\epsilon$~\cite{archer2011improved} \\ \hline
 \textbf{NW-RBP} & $(1+\epsilon, O(\frac{\log{n}}{\epsilon^2}))$~\cite{bateni2018improved, konemann2013lmp, moss2007approximation} \\ \hline
 \textbf{E-URAT} & $2$~\cite{paul2020budgeted} \\ \hline
 \textbf{NW-RQP} & $O(\log{n})$~\cite{bateni2018improved, konemann2013lmp, moss2007approximation} (tight)\\ \hline
 \textbf{EW-RQP} & $2$~\cite{garg2005saving, johnson2000prize} \\ \hline
\end{tabular}
\caption{A summary of the best bounds on some variants of prize collecting problems.}\label{tbBestBounds}
\end{center}
\end{table}

\HG{Very good table, Thank you.}

\end{comment}
%\input{Contribution/DynamicProgramming/banzafindex}
%\input{Problem/coalitiongames}
%\input{Problem/liquidgames}

\section{Notation and problem statement}

%Following Ghuge and Nagarajan~\cite{ghuge2020quasi}, we consider the following problem. 
For an integer $k$, let $[k]:=\{1,\ldots,k\}$. A directed \textit{path} is a directed graph made of a sequence of distinct vertices $(v_1, \dots, v_k)$ and a sequence of directed edges $(v_i, v_{i+1})$, $i\in[k-1]$.
An \emph{out-tree} (a.k.a. out-arborescence) is a directed graph in which there is exactly one directed path from a specific vertex $r$, called \emph{root}, to each other vertex. If a subgraph $T$ of a directed graph $D$ is an out-tree, then we say that $T$ is an out-tree of $D$. 

Let $D=(V, A)$ be a directed graph with $n$ nodes, $c:V \rightarrow \mathbb{Z}^{+}$ be a cost function on nodes,  $p:2^V \rightarrow \mathbb{R}^+\cup \{0\}$ be a monotone submodular prize function on the subsets of nodes, $r\in V$ be a root vertex, and $B$ be an integer budget. For any subgraph $D'$ of $D$, we denote by $V(D')$ and $A(D')$ the set of nodes and edges in $D'$, respectively. Given $S\subseteq V$, we denote the cost of $S$ by $c(S) = \sum_{v\in S}c(v)$ and we use shortcuts $c(D')=c(V(D'))$ and $p(D')=p(V(D'))$ for a subgraph $D'$ of $D$.
In the Directed Rooted Submodular Tree problem (\PName), the goal is to find an out-tree $T$ of $D$ rooted at $r$ such that $c(T)\le B$ and $p(T)$ is maximum. %We call this problem Directed Rooted Connected Out-subtree (\PName).
Throughout the paper, we denote an optimal solution to \PName by $T^*$.
%The cost of a path $P$, $c(P)$, is the sum of the costs of its nodes, $c(P)=\sum_{v_i \in V(P)} c(v_i)$. 
 
Given two nodes $u$ and $v$ in $V$, a path in $D$ from $u$ to $v$ with the minimum cost is called a \textit{shortest path} and its cost, denoted by $dist(u,v)$, is called the \textit{distance} from $u$ to $v$ in $D$.

%Throughout the paper, we denote by $T^*$ the optimal solution to \PName, where $T^*$ rooted at $r$.
%Recall that the undirected version of \textbf{EW-DRBP}, called \textbf{EW-RBP} has been investigated by Paul et al.~\cite{paul2020budgeted}. %in which we are given an undirected graph and the task is to find a tree $T^*$ containing a root vertex $r$ with maximum weight $\omega(T^*)$ such that $c(T^*) \le B$. To be consistent with the context, we term this problem \textbf{UTO}.\footnote{Note that Paul et al.~\cite{paul2020budgeted} called this problem \textit{budgeted prize-collecting minimum spanning tree}.} Paul et al.~\cite{paul2020budgeted} proposed a 2-approximation algorithm for \textbf{UTO}. 
\iffalse
\begin{center}
\noindent\fbox{\parbox{13.7cm}{
		\emph{Problem}: \PName\\
		\emph{Input}: a directed graph $ D= (V, A)$, a prize function $p:2^V \rightarrow \mathbb{R}^+\cup \{0\}$, a cost function $c: V \rightarrow \mathbb{Z}^+$, a specific vertex $r \in V$ and a budget $B\in\mathbb{N}$.\\
		\emph{Feasible Solution}: an out-tree $T=(V(T), A(T))$ rooted at  $r$ s.t. $c(T)\le B$. \\
		\emph{Goal}: $p(T)$ maximized.
}}\\
\end{center}
\fi

An algorithm is a bicriteria $(\beta, \alpha)$-approximation algorithm for \PName if, for any instance $I$ of the problem, it returns a solution $Sol_{I}$ such that $p(Sol_{I}) \ge \frac{OPT_I}{\alpha}$ and $c(Sol_{I}) \le \beta B$, where $OPT_I$ is the optimum for $I$.
%\gd{add unrooted version?}
%Let $OPT$ be the optimum to a maximization problem \textbf{BP} with a constraint on a budget problem and let be an algorithm $\mathcal{A}$ for \textbf{BP}.
%Let $Sol_{\mathcal{A}}$ be the solution returned by an algorithm $\mathcal{A}$ to \textbf{BP}. 
%For $\alpha,\beta\geq 1$, we say that algorithm $\mathcal{A}$ is a bicriteria $(\beta, \alpha)$-approximation for \textbf{BP} if $p(Sol_{\mathcal{A}}) \ge \frac{OPT}{\alpha}$ and $c(Sol_{\mathcal{A}}) \le \beta B$, for each solution $Sol_{\mathcal{A}}$ returned by $\mathcal{A}$. Let $OPT'$ be the optimum solution to a quota problem \textbf{QP}. 
%\gd{We can maybe avoid the following or specify it in the section on quota problems.}
%Let $Sol_{\mathcal{A}'}$ be the solution returned by an algorithm $\mathcal{A}'$ to \textbf{QP}. For $\alpha,\beta\geq 1$, we say that algorithm $\mathcal{A}'$ is a bicriteria $(\beta, \alpha)$-approximation for \textbf{QP} if $c(Sol_{\mathcal{A}'}) \le \alpha OPT$ and $p(Sol_{\mathcal{A}'}) \ge \frac{Q}{\beta}$. We omit the budget violation factor when this is equal to 1.
\section{Results and Techniques}

Our main result is given in the next theorem.
\begin{restatable}{theorem}{MainTheorem}\label{thMain}
\PName admits a polynomial-time bicriteria $\left(1+\epsilon, O\left(\frac{\sqrt{B}}{\epsilon^3}\right)\right)$-approximation algorithm, for any $\epsilon \in (0, 1]$.% that violates the budget constraint by a factor of $1+\epsilon$.% when for any $v \in V$, $c(v)=1$.
\end{restatable}

Our approach combines and extends techniques given by Kuo et al.~\cite{kuo2014maximizing} and Bateni et al.~\cite{bateni2018improved}.
%for \textbf{MCSB} (stated in the Related Work section).
%\textbf{MCSB} is an undirected version of \PName in which no distinguished node is required to be in the solution.
To illustrate our techniques, we now consider the case in which costs are unitary, i.e. $c(v)=1$, for each $v\in V$, and the prize function is additive, i.e. $p(S)=\sum_{v\in S}p(\{v\})$, for any $S\subseteq V$. In this case, the distance from a node $u$ to a node $v$ is equal to the minimum number of nodes in a path from $u$ to $v$ and the cost of a tree $T$ is equal to its size, $c(T)=|V(T)|$. W.lo.g. we also assume that the distance from $r$ to any node is at most $B$. We will give the proof for the general case in Section~\ref{sec:apxalgo}.

The algorithm works as follows. For any vertex $u$, we denote as $V_u$ the set of all nodes that are at a distance no more than $\lfloor \sqrt{B} \rfloor$ from $u$, $V_u:=\{v~|~dist(u,v)\leq \lfloor \sqrt{B} \rfloor\}$. We first select a subset $S_u$ of $V_u$ of at most $\lfloor \sqrt{B} \rfloor$ nodes  with the maximum prize, $S_u : = \arg\max\{ p(S)~:~ S\subseteq V_u, |S|\leq \lfloor \sqrt{B} \rfloor \}$.\footnote{This step can be done in polynomial time since function $p$ is additive.
% and all the costs are equal to 1. If the costs are general integers and the price is additive, this step consists in solving a knapsack problem, 
If $p$ is monotone and submodular, this step consists in solving the submodular maximization problem. See  Section~\ref{sec:apxalgo} for more details.} 
We then compute a minimal inclusion-wise out-tree $T_u$ rooted at $u$ that spans all nodes in $S_u$. Note that $|V(T_u)|\leq B$ since the distance from $u$ to any node in $S_u$ is at most $\lfloor \sqrt{B} \rfloor$. Let $z$ be a node such that $p(T_{z})$ is maximum. If $z = r$, then we take $T_{z}$ as our solution, otherwise we compute a solution by adding to $T_{z}$ a shortest path $P$ from $r$ to $z$ and removing the edges in $A(T_{z})\setminus A(P)$ incoming the nodes in $V(T_{z})\cap V(P)$. Let $T$ be our solution and $T^*$ be an optimal solution.

We will prove (Lemma~\ref{lmClaimKuoExtension}) that any out-tree $\hat{T}$ can be covered by at most $N=O(|\hat{T}|/m)$ out-subtrees $\{\hat{T}_i\}_{i=1}^N$ with at most $m$ nodes each, where $m$ is any positive integer less than $|\hat{T}|$. By applying this claim to an optimal solution $T^*$ and by setting $m=\lfloor \sqrt{B} \rfloor$, we obtain 
\[
p(T^*) = p\left( \bigcup_{i=1}^N V(T^*_i)  \right) \leq N p(T'),
\]
where $p(T') = \max\{ p(T^*_i ) ~|~ i\in [N] \}$, $|T'|\leq \lfloor \sqrt{B} \rfloor$, and $N=O(|T^*|/m) = O(\sqrt{B})$. Let $w$ be the root of $T'$. 
Recall that $S_{w}$ is a set of at most $\lfloor \sqrt{B} \rfloor$ nodes that are at a distance no more than  $\lfloor \sqrt{B} \rfloor$  from $w$ and have the maximum prize and $T_w$ contains all the nodes in $S_w$. Since $|T'|\leq \lfloor \sqrt{B} \rfloor$, we have
\[
p(T') \leq p(S_w) \leq p(T_{w}) \leq p(T_{z}) \leq p(T).
\]
Since $N= O(\sqrt{B})$, we conclude that $p(T^*) = O(\sqrt{B}) p(T)$.
%we obtain $N$ out-trees $\{T^*_i\}_{i=1}^N$ and

Note that the cost of $T$ is upper-bounded by $2B$, as both the cost of $T_{z}$ and that of a shortest path from $r$ to $z$ are at most $B$. We can use the trimming procedure introduced by Bateni et al.~\cite{bateni2018improved} to obtain an out-subtree of $T$ with cost at most $(1+\epsilon)B$ by loosing an approximation factor of $O(1/\epsilon^2)$, for any $\epsilon\in (0,1]$ (see Lemma~\ref{coTrimmingProcess}). This shows Theorem~\ref{thMain} for the unit-cost, additive-prize case. In the case in which the prize is a general monotone submodular function, the trimming procedure by Bateni et al. cannot be applied. We show how to generalize this procedure to the case of any  monotone submodular prize function by loosing an extra approximation factor of $O(1/\eps)$.

%We find an out-tree rooted at $u$ and connected to the chosen nodes. This process is repeated for any vertex $u$ and the out-tree with the maximum prize is selected. If $r$ is the root of the chosen tree, we are done. Otherwise, a directed path from $r$ to the root of tree would be added to the final solution, which may cost an extra budget at most $B$. The key step in the analysis of the algorithm to show that the approximation ratio is equal to $O(\sqrt{B})$. This can be shown by the main claim which states that any tree of cost $c(T)$ can be covered by at most $O(c(T)/m)$ subtrees such that each subtree has no cost more than $c(r_i)+m$, where $r_i$ is the root of each subtree $T^i$. 

We can use the same approach to obtain a polynomial-time bicriteria $\left(1+\epsilon, O\left(\frac{\sqrt{B}}{\epsilon^2}\right)\right)$-approximation algorithm for the case of additive prize function and edge-cost. 
More importantly, we can obtain a polynomial-time bicriteria $\left(1+\eps, O\left(\frac{\sqrt{B}}{\epsilon^3}\right)\right)$-approximation algorithm for \textbf{STO}, i.e. for the edge-cost case where the prize function is monotone submodular. To the best of our knowledge, this is the first polynomial-time approximation algorithm for \textbf{STO}. 

Finally, for the unrooted version the same approach with some minor changes achieves an $O(\sqrt{B})$-approximation with no budget violation. 

\section{Approximation Algorithm for \PName}\label{sec:apxalgo}

%\subsection{A Polynomial-Time Approximation Algorithm}

We now introduce our polynomial-time approximation algorithm for \PName. We start by defining a procedure that takes as input an out-tree of a directed graph $D$ and returns another out-tree of $D$  which has a smaller cost but preserves the same prize-to-cost ratio (up to a bounded multiplicative factor). 
%outputs another out-tree with the same root and some properties on the ratio between prize and cost.

Bateni et al.~\cite{bateni2018improved} introduced a similar procedure for the case of undirected graphs and additive prize function.
%procedure to compute a subtree of a rooted node-weighted tree (\ED{the same thing here}) which has the same (up to a constant factor) ratio between prize and cost of the original tree but has a smaller cost.
In their case, we are given an undirected graph $G=(V, E)$, a distinguished vertex $r \in V$ and a budget $B$, where each vertex $v \in V$ is assigned with a prize $p'(v)$ and a cost $c'(v)$. For a tree $T$, the prize and cost of $T$ are the sum of the prizes and costs of the nodes of $T$ and are denoted by $p'(T)$ and $c'(T)$, respectively.  A graph $G$ is called \emph{$B$-proper} for the vertex $r$ if the cost of reaching any vertex from $r$ is at most $B$. Bateni et al. proposed a trimming process that leads to the following lemma.
%Let $T$ be a tree, where $V(T) \subseteq V$ and $E(T) \subseteq E$ and let  $c'(T)=\sum_{v \in V(T)} c'(v)$ and $p'(T)=\sum_{v \in V(T)} p'(v)$. The trimming procedure by Bateni et al. takes  as input an out-subtree $T$ rooted at a node $r$ of a $B$-proper graph $G$ and first (i) computes an out-subtree $T''$ of $T$, not necessarily rooted $r$, such that $\frac{\epsilon B}{2}\leq c'(T'') \leq \eps B$ and $p'(T'')\geq \frac{\eps B}{2}\gamma$, where $\gamma = \frac{p'(T)}{c'(T)}$. Then (ii) connects node $r$ to the root of $T''$ with a shortest path and obtain a tree $T'$ rooted at $r$. Since $G$ is $B$-proper, this path exists and has length at most $B$, which implies that the cost of the resulting tree is between $\frac{\epsilon B}{2}$ and $(1+\eps)B$, and its prize to cost ratio is at least $\frac{p'(T'')}{(1+\eps)B}\geq\frac{\epsilon \gamma}{4}$. This leads to the following lemma.
%Let $\gamma=\frac{p'(T)}{c'(T)}$ be the prize-to-cost ratio of $T$. Bateni, Hajiaghayi and Liaghat~\cite{bateni2018improved} proposed a trimming process that leads to the following.
%
\begin{lemma}[Lemma 3 in \cite{bateni2018improved}]\label{lmBateniTrimmingProcess}
Let $T$ be a tree rooted at $r$ with the prize-to-cost ratio $\gamma=\frac{p'(T)}{c'(T)}$. Suppose the underlying graph is $B$-proper for $r$ and for $\epsilon \in (0, 1]$ the cost of the tree is at least $\frac{\epsilon B}{2}$. One can find a tree $T'$ containing $r$ with the prize-to-cost ratio at least $\frac{\epsilon \gamma}{4}$ such that $\epsilon B/2 \le c'(T') \le (1+\epsilon)B$.
\end{lemma}
We now generalize this trimming process to the case in which the underlying graph is directed and the prize function is monotone and submodular by borrowing ideas from~\cite{bateni2018improved}. 

We introduce some additional definitions. Let $T$ be an out-tree rooted at $r$. A \emph{full} out-subtree of $T$ rooted at some node $v$ is an out-subtree of $T$ containing all the vertices that are reachable from $r$ through $v$ in $T$. The set of \emph{strict} out-subtrees of $T$ is the set of all full out-subtrees of $T$ other than $T$ itself. The set of \emph{immediate} out-subtrees of $T$ is the set of all full out-subtrees rooted at the children of $r$ in $T$. A directed graph $D=(V,A)$ is $B$-\emph{appropriate} for a node $r$ if $dist(r,v)\leq B$ for any node $v \in V$.

\begin{lemma}\label{lmBateniTrimmingProcessGeneral}
Let $D= (V, A)$ be a $B$-appropriate graph for a node $r$. Let $T$ be an out-tree of $D$ rooted at $r$ with the prize-to-cost ratio $\gamma=\frac{p(T)}{c(T)}$, where $p$ is a monotone submodular function. Suppose that $\frac{\epsilon B}{2}\le c(T) \le hB$, where $h \in (1, n]$ and $\epsilon \in (0, 1]$. One can find an out-subtree $\hat{T}$ rooted at $r$ with the prize-to-cost ratio at least $\frac{\epsilon^2 \gamma}{32h}$ such that $\epsilon B/2 \le c(\hat{T}) \le (1+\epsilon)B$.
\end{lemma}
\begin{proof}

We run the following initial trimming procedure. We iteratively remove a strict out-subtree $T'$ from $T$ that satisfies two conditions: (i) the prize-to-cost ratio of $T\setminus T'$ is at least $\gamma$, and (ii) $c(T\setminus T') \ge \frac{\epsilon}{2}B$. We repeat this process until no such strict out-subtree exists. Let $T_{-}$ be the remaining out-tree after applying this process on $T$.

Now if $c(T_{-}) \le (1+\epsilon) B$, the desired out-subtree is obtained. Suppose it is not the case. A full out-subtree $T'$ is called \emph{rich} if $c(T') \ge \frac{\epsilon}{2} B$ and the prize-to-cost ratio of $T'$ and all its strict out-subtrees are at least $\gamma$. We claim that if there exists a rich out-subtree, then we can find the desired out-subtree $\hat{T}$.

\begin{claim}\label{clBateniRichGeneral}
Given a rich out-subtree $T'$, the desired out-subtree $\hat{T}$ can be found.
\end{claim}

\begin{proof}
We first find a rich out-subtree $T''$ of $T'$ such that the strict out-subtrees of $T''$ are not rich, i.e., $c(T'') \ge \frac{\epsilon}{2}B$ while the cost of any strict out-subtree of $T''$ (if any exist) is less than $\frac{\epsilon}{2}B$. Let $C$ be the total cost of the immediate out-subtrees of $T''$. We distinguish between two cases:

\begin{enumerate}
    \item If $C < \frac{\epsilon}{2}B$, then let $\hat{T}$ be the union of $T''$ and a shortest path $P$ from $r$ to the root $r''$ of $T''$. $\hat{T}$ has cost at most $C+B \le (1+\epsilon)B$ and prize at least $\gamma(\frac{\epsilon}{2}B)$. This implies that $\hat{T}$ has ratio at least $\frac{\gamma \epsilon}{2(1+\epsilon)} \ge \frac{\gamma\epsilon}{4}\ge \frac{\gamma \epsilon^2}{32h}$. %Note that $T''$ is rich, implying that it has the total prize at least $\gamma(\frac{\epsilon}{2}B)$.
    
    \item If $C \ge \frac{\epsilon}{2}B$, we proceed as follows. Since each immediate out-subtree of $T''$ has a cost strictly smaller than $\frac{\epsilon}{2}B$, we can partition all the immediate out-subtrees of $T''$ into $M$ groups $S_1, \dots, S_M$ in such a way that for each $i\in [M-1]$ the total cost of immediate out-subtrees in $S_i$ is at least $\frac{\epsilon}{2}B$, and for each $i\in[M]$ it is at most $\epsilon B$. 
    %We set $S_M=T'' \setminus \bigcup_{i=1}^{M-1}S_i$. 
    We can always group in this way since the cost of each immediate out-subtree of $T''$ is less than $\frac{\epsilon}{2}B$ while $C\ge \frac{\epsilon}{2}B$. Since the total cost of all the immediate out-subtrees of $T''$ is upper bounded by $hB$, then the number of selected groups $M$ is at most
    \[
    M \le \left\lceil \frac{hB}{\frac{\epsilon}{2}B} \right\rceil = \left\lceil \frac{2h}{\epsilon} \right\rceil \le \left\lfloor \frac{2h}{\epsilon}\right\rfloor +1\le \left\lfloor \frac{4h}{\epsilon}\right\rfloor \le \frac{4h}{\epsilon}.
    \]
     
     We now add the root $r''$ of $T''$ to each group $S_i$ and denote the new group by $S'_i$, i.e., $S'_i=S_i \cup \{r''\}$, for any $i \in [M]$. By the monotonicity and submodularity of $p$, we have $\sum_{i=1}^{M}p(S'_i)\ge p(S'_1) + \sum_{i=2}^{M}p(S_i) \ge p(S'_1 \cup\bigcup_{i=2}^{M} S_i)=p(T'')$. Now among $S'_1, \dots, S'_M$, we select the group $S'_{z}$ that maximizes the prize, i.e., $z=\arg\max_{i \in [M]}p(S'_i)$. We know that
    \[
    p(S'_{z}) \ge\frac{1}{M}\sum_{i=1}^{M}p(S'_i)\ge \frac{p(T'')}{M}\ge\frac{\epsilon}{4h} p(T'') \ge \frac{\epsilon}{4h} \cdot \frac{\gamma \epsilon}{2}B=\frac{\gamma \epsilon^2}{8h}B.
    \]
    In case $z=M$ and $c(S'_M) <\frac{\epsilon}{2}B$, we select a subset of immediate out-subtrees from $\bigcup_{i=1}^{M-1}S_i$ with the total cost of at least $\frac{\epsilon}{2}B$ and at most $\epsilon B-c(S'_M)$, and add it to $S'_z$. 
    
    Finally, let $\hat{T}$ be the union of a shortest path $P$ from $r$ to $r''$, $S'_{z}$, and the edges from $r''$ to the roots of the out-subtrees in $S_{z}$ (see Figure~\ref{figRichOut-subtree}). By monotonicity, $\hat{T}$ has the total prize at least $p(\hat{T}) \ge p(S'_{z})\ge \frac{\gamma \epsilon^2}{8h}B$. Note that $c(\hat{T}) \le (1+\epsilon)B$ as $c(S'_{z} \setminus \{r''\}) =c(S_{z}) \le \epsilon B$ and the shortest path from $r$ to $r''$ costs at most $B$ (since the graph is $B$-appropriate). This implies that the prize-to-cost ratio of $\hat{T}$ is at least $\frac{\gamma \epsilon^2}{8h(1+\epsilon)} \ge \frac{\gamma \epsilon^2}{16h}\ge \frac{\gamma \epsilon^2}{32h}$.\qedhere

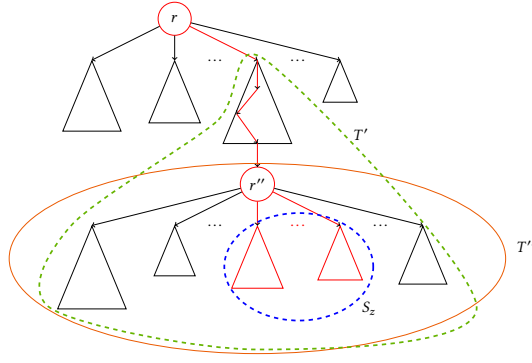
\begin{figure}[t]
\centering
\scalebox{0.55}{
\begin{tikzpicture}
\node at (2,-2) [circle,draw=red,minimum size=8mm] (r) {$r$};
%\node at (0,-4) [circle,draw,minimum size=8mm] (B) {$\beta$};
\node[coordinate] at (0,-3) (A) {};
\node at (0,-3) [isosceles triangle, shape border rotate=+90,
draw,minimum size=8mm,minimum height=1.7cm,
anchor=north] (Ctriangle) {};
\node[coordinate] at (2,-3) (B) {};
\node at (2,-3) [isosceles triangle, shape border rotate=+90,
draw,minimum size=8mm,minimum height=1.5cm,
anchor=north] (Ctriangle) {};
\node[coordinate] at (4,-3) (C) {};
\node at (4,-3) [isosceles triangle, shape border rotate=+90,
draw,minimum size=8mm,minimum height=2cm,
anchor=north] (Ctriangle) {};
\node[coordinate] at (6,-3) (D) {};
\node at (6,-3) [isosceles triangle, shape border rotate=+90,
draw,minimum size=8mm,minimum height=1cm,
anchor=north] (Ctriangle) {};

\node[coordinate] at (4,-5) (Dl) {};

%\node at (4,-6) [circle,draw,minimum size=4mm] (Dl') { };
\node at (4,-6) [circle,draw=red,minimum size=4mm] (r') {$r''$};
\node[coordinate] at (0,-7) (Ar') {};
\node at (0,-7) [isosceles triangle, shape border rotate=+90,
draw,minimum size=8mm,minimum height=2cm,
anchor=north] (Ctriangle) {};
\node[coordinate] at (2,-7) (Br') {};
\node at (2,-7) [isosceles triangle, shape border rotate=+90,draw,minimum size=8mm,minimum height=1.2cm,anchor=north] (Ctriangle) {};
\node[coordinate] at (4,-7) (Cr') {};
\node at (4,-7) [isosceles triangle, shape border rotate=+90,draw=red,minimum size=8mm,minimum height=1.5cm,anchor=north] (Ctriangle) (S1) {};
\node[coordinate] at (6,-7) (Dr') {};
\node at (6,-7) [isosceles triangle, shape border rotate=+90,draw=red,minimum size=8mm,minimum height=1.3cm,anchor=north] (Ctriangle) (S2) {};
\node[coordinate] at (8,-7) (Er') {};
\node at (8,-7) [isosceles triangle, shape border rotate=+90,draw,minimum size=8mm,minimum height=1.4cm,anchor=north] (Ctriangle) {};

\draw [->] (r) to (A);
\draw [->] (r) to (B);
\draw [->, draw=red] (r) to (C);
\path (B) to node {\dots} (C);
\draw [->] (r) to (D);
\path (C) to node {\dots} (D);
\draw [->, draw=red] (Dl) to (r');
\draw [->] (r') to (Ar');
\draw [->] (r') to (Br');
\path (Br') to node {\dots} (Cr');
\draw [->, draw=red] (r') to (Cr');
\draw [->, draw=red] (r') to (Dr');
\draw [->, draw] (r') to (Er');

\path (Cr') to node {\textcolor{red}{\dots}} (Dr');
\path (Dr') to node {\dots} (Er');

\begin{scope}[every path/.style={->}]
\path [draw = red, inner sep=100pt] (4,-3) -- (4,-3.7);
%\path [draw = red, inner sep=100pt] (4,-3.5) -- (4,-3.7);
\path [draw = red, inner sep=100pt] (4,-3.7) -- (3.5,-4.3);
\path [draw = red, inner sep=100pt] (3.5,-4.3) -- (4,-5);
\end{scope}

%\node [rotate=-60][draw,dashed,inner sep=1pt, circle,yscale=1, fit={(S1) (S2)}] {};
%\draw [very thick] plot [smooth cycle] coordinates {(4,-7) (4,-7) (6,-8) (6,-8) (3,-7) };
\draw[dashed, very thick, draw=blue] (5,-8) ellipse (1.8cm and 1.3cm);
\node at (6.7,-9) (S_z) {$S_z$};

\draw[ draw=orange] (4,-7.8) ellipse (6cm and 2.3cm);
\node at (10.5,-7.5) (T'') {$T''$};
%\draw [draw=orange, very thick] plot [smooth cycle] coordinates {(4.2,-5.4) (-1,-7) (-1,-9) (4.2,-9) (7,-9) (9.3,-8)};

\draw [dashed, very thick, draw=green] plot [smooth cycle] coordinates {(4.2,-3) (3,-4.5) (-.5,-7) (-1,-9) (4,-10) (9.3,-9)};
\node at (6.5,-4.8) (T') {$T'$};

\end{tikzpicture}}
\caption{$T_{-}$ is rooted at $r$, which is the whole out-tree. The green dashed closed curve represents the rich out-subtree $T'$. The orange circle represents $T''$ rooted at $r''$, where its strict out-subtrees are not rich, i.e., the cost of any strict out-subtree of $T''$ is less than $\frac{\epsilon}{2}B$. The blue dashed circle represents the partition $S_z$, which maximizes the prize and costs at most $\epsilon B$. The red out-subtree represents $\hat{T}$, which is the union of a shortest path from $r$ to $r''$, $S_z$ and the edges from $r''$ to the immediate out-subtrees of $T''$ in $S_z$. Note that for the sake of simplicity, in this figure we suppose that the shortest path from $r$ to $r''$ is included in $T_{-}$.}\label{figRichOut-subtree}

\end{figure}

\end{enumerate}
\end{proof}

It only remains to consider the case when there is no rich out-subtree. Since $T_{-}$ is not rich and $c(T_{-}) \ge \frac{\epsilon}{2}B$, the ratio of at least one strict out-subtree of $T_{-}$ is less than $\gamma$. Now we find a strict out-subtree $T'$ with ratio less than $\gamma$ such that the ratio of all of its strict out-subtrees (if any exist) is at least $\gamma$. %Note that such an out-subtree exists as if we have an out-subtree $T''$ with ratio less than $\gamma$ which has an out-subtree $T'''$ with ratio less than $\gamma$, we can select $T'''$ instead.
We first need to show that $c(T_{-}\setminus T') < \frac{\epsilon}{2}B$. 

\begin{claim}\label{clLowCost}
 $c(T_{-}\setminus T') < \frac{\epsilon}{2}B$.
\end{claim}

\begin{proof}

By the submodularity of $p$, we know that $p(T_{-}\setminus T')+p(T') \ge p(T_{-})$. This implies that 
\begin{align}\label{eqSubmodularityNotRich}
    \frac{p(T_{-} \setminus T')}{c(T_{-} \setminus T')} \ge \frac{p(T_{-})-p(T')}{c(T_{-})-c(T')}.
\end{align}

Let $\gamma'=\frac{p(T')}{c(T')}$ be the prize-to-cost ratio of $T'$. We know that 
\begin{align}\label{eqRatioNotRich}
   p(T_{-})-p(T')= c(T_{-})\gamma -c(T') \gamma' > c(T_{-})\gamma -c(T') \gamma,
\end{align}
where the inequality holds because $\gamma' < \gamma$. By Equations~\eqref{eqSubmodularityNotRich} and~\eqref{eqRatioNotRich}, we have $\frac{p(T_{-} \setminus T')}{c(T_{-} \setminus T')} > \gamma$. As the prize-to-cost ratio of $T_{-} \setminus T'$ is more than $\gamma$ but $T'$ has not been removed from $T$ during the initial phase, then $c(T_{-}\setminus T') < \frac{\epsilon}{2}B$. This concludes the proof of the claim. 
\end{proof}

We know that $c(T_{-})> (1+\epsilon)B$ and the cost from $r$ to the root of $T'$ is at most $B$. Then by Claim~\ref{clLowCost}, the total cost of immediate out-subtrees of $T'$ is at least $\frac{\epsilon}{2}B$. Also, the cost of an immediate out-subtree of $T'$ is less than $\frac{\epsilon}{2}B$, otherwise, we have a rich out-subtree. As the ratio and cost of $T_{-}$ are at least $\gamma$ and $\frac{\epsilon}{2}B$, respectively, then $p(T_{-}) \ge \frac{\gamma \epsilon}{2}B$. Now we distinguish between two cases:

\begin{enumerate}
    \item  If $p(T') \ge \frac{\gamma \epsilon}{4}B$, by similar reasoning as above, we group the immediate out-subtrees of $T'$ into $M$ groups $S_1, \dots, S_M$ in such a way that for each $i\in [M-1]$ the total cost of immediate out-subtrees in $S_i$ is at least $\frac{\epsilon}{2}B$, and for each $i\in[M]$ it is at most $\epsilon B$. Now define a new group $S'_i=S_i \cup \{r'\}$, for any $i \in [M]$. Let $z=\arg\max_{i \in [M]}p(S'_i)$. Then the group $S'_{z}$, which maximizes the prize is selected. We know that $M \le \frac{4h}{\epsilon}$. Hence, $p(S'_{z}) \ge \frac{\epsilon}{4h}p(T') \ge  \frac{\epsilon}{4h}\cdot \frac{\gamma \epsilon}{4}B=\frac{\gamma \epsilon^2}{16h}B.$
    
    Note that in case $z=M$ and $c(S'_M) <\frac{\epsilon}{2}B$, we select a subset of immediate out-subtrees from $\bigcup_{i=1}^{M-1}S_i$ with the total cost of at least $\frac{\epsilon}{2}B$ and at most $\epsilon B-c(S'_M)$, and add it to $S'_z$. 
    
    Let $\hat{T}$ be the union of a shortest path $P$ from $r$ to $r'$, $S'_{z}$, and the edges from $r'$ to the roots of the out-subtrees in $S_{z}$. The cost of $\hat{T}$ is at most $(1+\epsilon)B$ and the prize-to-cost ratio is at least $\frac{\gamma \epsilon^2}{16h(1+\epsilon)} \ge \frac{\gamma \epsilon^2}{32h}$.

    \item If $p(T') < \frac{\gamma \epsilon}{4}B$, we proceed as follows. Consider the out-subtree $T''=T_{-}\setminus T'$, which is rooted at $r$. Recall that by Claim~\ref{clLowCost}, we have $c(T'') <\frac{\epsilon}{2}B$. We connect a subset of immediate out-subtrees $T'_1, \dots, T'_q$ of $T'$ with cost $\frac{\epsilon}{2}B-c(T'')\le c(\bigcup_{i=1}^{q}T'_i) \le \epsilon B-c(T'')$ to the root of $T''$ through the root of $T'$.
    Since the cost of each immediate out-subtree of $T'$ is less than $\frac{\epsilon}{2}B$ (otherwise, we have a rich out-subtree) and $c(T') >(1+\frac{\epsilon}{2})B$, a subset of immediate out-subtrees $T'_1, \dots, T'_q$ of $T'$ with such a cost can be found.
    We call the  resulting out-subtree $\hat{T}$ and observe that $c(\hat{T})\ge \frac{\epsilon}{2}B$ (see Figure~\ref{figLightOut-subtree}). We now bound the prize-to-cost ratio of $\hat{T}$. First note that by the submodularity of $p$, $p(T'')+p(T') \ge p(T_{-})$. Thus by the subcase assumption and the monotonicity of $p$, we have $p(\hat{T}) \ge p(T'') \ge \frac{\gamma \epsilon}{4}B$. Since $\frac{\epsilon}{2}B-c(T'') \le c(\bigcup_{i=1}^{q}T'_i) \le \epsilon B-c(T'')$ and the graph is $B$-appropriate, $c(\hat{T}) \le (1+\epsilon)B$. Therefore, the prize-to-cost ratio of the resulting out-subtree $\hat{T}$ is $\frac{\gamma \epsilon}{4(1+\epsilon)} \ge \frac{\gamma\epsilon}{8} \ge \frac{\gamma \epsilon^2}{32h}$.

\begin{figure}[t]
\centering
\scalebox{0.55}{
\begin{tikzpicture}
\node at (2,-2) [circle,draw=red,minimum size=8mm] (r) {$r$};
%\node at (0,-4) [circle,draw,minimum size=8mm] (B) {$\beta$};
\node[coordinate] at (0,-3) (A) {};
\node at (0,-3) [isosceles triangle, shape border rotate=+90,
draw=red,minimum size=8mm,minimum height=1.7cm,
anchor=north] (Ctriangle) {};
\node[coordinate] at (2,-3) (B) {};
\node at (2,-3) [isosceles triangle, shape border rotate=+90,
draw=red,minimum size=8mm,minimum height=1.5cm,
anchor=north] (Ctriangle) {};
\node[coordinate] at (4,-3) (C) {};
\node at (4,-3) [isosceles triangle, shape border rotate=+90,
draw=red,minimum size=8mm,minimum height=2cm,
anchor=north] (Ctriangle) {};
\node[coordinate] at (6,-3) (D) {};
\node at (6,-3) [isosceles triangle, shape border rotate=+90,
draw=red,minimum size=8mm,minimum height=1cm,
anchor=north] (Ctriangle) {};

\node[coordinate] at (4,-5) (Dl) {};

\node at (4,-6) [circle,draw=red,minimum size=4mm] (r') {$r'$};
\node[coordinate] at (0,-7) (Ar') {};
\node at (0,-7) [isosceles triangle, shape border rotate=+90,
draw=red,minimum size=8mm,minimum height=2cm,
anchor=north] (Ctriangle) {};
\node[coordinate] at (2,-7) (Br') {};
\node at (2,-7) [isosceles triangle, shape border rotate=+90,draw=red,minimum size=8mm,minimum height=1.2cm,anchor=north] (Ctriangle) {};
\node[coordinate] at (4,-7) (Cr') {};
\node at (4,-7) [isosceles triangle, shape border rotate=+90,draw,minimum size=8mm,minimum height=1.5cm,anchor=north] (Ctriangle) (S1) {};
\node[coordinate] at (6,-7) (Dr') {};
\node at (6,-7) [isosceles triangle, shape border rotate=+90,draw,minimum size=8mm,minimum height=1.3cm,anchor=north] (Ctriangle) (S2) {};

\draw [->, draw=red] (r) to (A);
\draw [->, draw=red] (r) to (B);
\draw [->, draw=red] (r) to (C);
\path (B) to node {\textcolor{red}{\dots}} (C);
\draw [->, draw=red] (r) to (D);
\path (C) to node {\textcolor{red}{\dots}} (D);
\draw [->, draw=red] (Dl) to (r');
\draw [->, draw=red] (r') to (Ar');
\draw [->, draw=red] (r') to (Br');
\path (Br') to node {\dots} (Cr');
\draw [->] (r') to (Cr');
\draw [->] (r') to (Dr');

\path (Ar') to node {\textcolor{red}{\dots}} (Br');
\path (Cr') to node {\dots} (Dr');

\draw[dashed, very thick, draw=blue] (.6,-8) ellipse (2.2cm and 1.5cm);
\node at (0,-6.3) (T'_1) {$\bigcup_{i=1}^{q}T'_i$};

\draw[draw=orange] (2.5,-7.5) ellipse (5cm and 2.3cm);
\node at (8,-7) (T') {$T'$};

\draw[dashed, draw=green] (2.5,-3.1) ellipse (5cm and 2.3cm);
%\node[coordinate] at (8,-2) (T'') {$T''$};
\node at (7.5,-2) (T'') {$T''$};
%\draw [dashed, very thick, draw=green] plot [smooth cycle] coordinates {(4.2,-3) (3,-4.5) (-.5,-7) (-1,-9) (4,-10) (9.3,-9)};
\end{tikzpicture}}
\caption{$T_{-}$ is rooted at $r$, which is the whole out-tree. The green dashed circle represents $T''$ rooted at $r$. The orange circle represents $T'$ rooted at $r'$, where its cost is more than $(1+\frac{\epsilon}{2})B$. The blue dashed circle represents a subset of immediate out-subtrees $T'_1, \dots, T'_q$ of $T'$ with cost $\frac{\epsilon}{2}B-c(T'')\le c(\bigcup_{i=1}^{q}T'_i) \le \epsilon B-c(T'')$. The red out-subtree represents $\hat{T}$, which is the union of $T''$, the edge from $T''$ to $r'$, $\bigcup_{i=1}^{q}T'_i$ and the edges from $r'$ to $T'_1, \dots, T'_q$.}\label{figLightOut-subtree}
\end{figure}
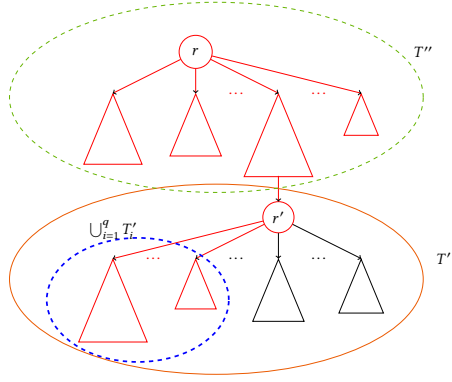

\end{enumerate}
The proof is complete.
\end{proof}
%Given a directed graph $D=(V,E)$. We first consider the case where for each $v \in V$, $c(v)=1$.
To propose our algorithm, we need a last element. Let $U=\{x_1, \dots, x_n\}$ be a ground set, $c:U \rightarrow Z^{+}$ be a cost function, $f:2^{U} \rightarrow \mathbb{R}^+ \cup \{0\}$ be a monotone submodular function, and $K$ be an integer budget. In the Submodular Maximization problem (\textbf{SM}), we are looking for a subset $S \subseteq U$ such that $|S|\le K$ and $f(S)$ is maximum. Nemhauser et al.~\cite{nemhauser1978analysis} provided a greedy algorithm that starts from $S:=\emptyset$ and runs $K$ iterations in which, at each iteration, it adds to $S$ the element $x$ which maximizes $f(S\cup\{x\}) - f(S)$. This algorithm guarantees a $(1-e^{-1})$-approximation for \textbf{SM}. We denote by \textbf{RSM} the rooted variant of \textbf{SM} in which, additionally, a specific element $v \in U$ is required to be included in the solution, that is we are looking for a subset $S \subseteq U$ such that $|S|\le K$, $v \in S$ and $f(S)$ is maximum.
We can run Nemhauser et al.~\cite{nemhauser1978analysis}'s approach for \textbf{RSM} with a minor change: we initialize $S:=\{v\}$ and run $K-1$ greedy iterations. We call this approach \textbf{Greedy}. It can be shown that \textbf{Greedy} guarantees a $(1-e^{-1})$-approximation algorithm for \textbf{RSM}  (see e.g.~\cite{kuo2014maximizing}).

\begin{algorithm}[t]
	\caption{}
	\label{algNW-DRBP}
	%\hspace*{\algorithmicindent} 
	\textbf{Input:} Directed graph $ D= (V, A)$; monotone submodular prize function $p: 2^V \rightarrow \mathbb{R}^+ \cup\{0\}$; cost function $c: V \rightarrow \mathbb{Z}^+$; root $r \in V$; budget $B$; and $\epsilon' \in (0, 1]$.\\
	%\hspace*{\algorithmicindent}
	\textbf{Output:} Out-tree $T$ of $D$ rooted at $r$ such that $c(T)\le (1+\epsilon')B$. 
	\vspace{-4mm}
	\begin{algorithmic}[1]%\baselineskip=14pt\relax
	    \STATE Remove from $D$ all the nodes at a distance more than $B$ from $r$;
	    \FOR{$u \in V$}
	      \STATE $V_u:=\{v~|~dist(u,v)\leq c(u)+ \lfloor\sqrt{B}\rfloor\}$;
	      \STATE Define an instance $I^{RSM}_u$ of \textbf{RSM} with elements $V_u$, specific element $u$, budget $\lfloor \sqrt{B} \rfloor+1$, profits $p(S)$, for each $S\subseteq V_u$;
	      \STATE Let $S_u$ be a $(1-e^{-1})$-approximate solution to $I^{RSM}_u$, computed by using  \textbf{Greedy};%the algorithm of Nemhauser et al.~\cite{nemhauser1978analysis} such that $u \in S_u$;
	      \STATE Let $T_u$ be a minimal inclusion-wise out-tree rooted at $u$ spanning all nodes in $S_u$;
	      %\IF{$p(S_u) \ge Pz$}
	      %\STATE $Prz=p(S_u)$;
	      %\STATE $z=u$;
	      %\ENDIF
	    \ENDFOR
        \STATE $z:=\arg\max_{u\in V} p(T_u)$;
        %\IF{$z=r$}
        % \STATE $T:=T_z$;
%       \ELSE
        \STATE Let $P$ be a shortest path from $r$ to $z$;
        \STATE $T:= P \cup T_z$;
        \STATE $A(T):= A(T) \setminus \{ (v,w)\in A(T_z)\setminus A(P)~:~ w\in V(T_{z})\cap V(P)  \}$;
%        \ENDIF
        %\STATE Let $\gamma=\frac{p(T)}{c(T)}$.
        \STATE Apply the trimming process in Lemma~\ref{lmBateniTrimmingProcessGeneral} with $\epsilon=\epsilon'$ to $T$;\label{lnTrimmingProcess}% to obtain another tree $T$ rooted at $r^*$ with prize-to-cost ratio $\frac{\epsilon \gamma}{4}$.
        \RETURN $T$.
	\end{algorithmic}
\end{algorithm}
Now we can propose our approximation algorithm for \PName, which is reported in Algorithm~\ref{algNW-DRBP}.
In words, Algorithm~\ref{algNW-DRBP} first computes the maximal inclusion-wise $B$-appropriate subgraph for $r$ of a given graph $D$ by removing all the nodes at a distance larger than $B$ from $r$. Let $D=(V,A)$ be the resulting directed graph. For each node $u$, it computes the set $V_u$ of all nodes that are at a distance no more than $c(u)+\lfloor \sqrt{B}\rfloor$ from $u$. Let $S^*_u$ be a subset of $V_u$ such that $|S^*_u|\leq \lfloor \sqrt{B}\rfloor+1$, $u \in S^*_u$, and $p(S^*_u)$ is maximum. Finding $S^*_u$ requires to solve an instance $I^{RSM}_u$ of \textbf{RSM} where the elements are $V_u$, the budget is $\lfloor \sqrt{B}\rfloor+1$, the specific element is $u$,  and profits are defined by function  $p(\cdot)$. %Using the algorithm by Nemhauser et al.~\cite{nemhauser1978analysis}, 
Using \textbf{Greedy}, Algorithm~\ref{algNW-DRBP} computes in polynomial time an approximate solution $S_u$ to $I^{RSM}_u$ with $u\in S_u$, $|S_u|\leq \lfloor\sqrt{B}\rfloor+1$ and $p(S_u)\geq (1-e^{-1}) p(S^*_u)$. Finally, for each $u \in V$, Algorithm~\ref{algNW-DRBP} computes a spanning out-tree $T_u$ rooted in $u$ that spans all the nodes in $S_u$. Let $z$ be a node such that $p(T_z)$ is maximum, i.e., $z=\arg\max_{u \in V} p(T_u)$. Then, we have  $c(T_z\setminus \{z\}) \le B$ as $|S_z\setminus \{z\}| \le \lfloor \sqrt{B}\rfloor$ and $dist(z, v) \le c(z)+\lfloor \sqrt{B}\rfloor$ for any $v \in S_z$. If $z=r$, then Algorithm~\ref{algNW-DRBP} defines $T=T_z$. Otherwise, it computes a shortest path $P$ from $r$ to $z$ and defines $T$ as the union of $T_z$ and $P$. Since the obtained graph might not be an out-tree, Algorithm~\ref{algNW-DRBP} removes the possible edges incoming the nodes in $V(T_{z})\cap V(P)$ that belong only to $T_z$. 
The obtained out-tree $T$ has a cost of at most $2B$ as $dist(r, z) \le B$ and $c(T_z\setminus \{z\}) \le B$. Therefore, Algorithm~\ref{algNW-DRBP} applies the trimming process in Lemma~\ref{lmBateniTrimmingProcessGeneral} to $T$ to reduce the cost to $(1+\eps)B$, where $\epsilon \in (0, 1]$ and outputs the resulting out-tree.

In the next theorem, we show that Algorithm~\ref{algNW-DRBP} guarantees a bicriteria approximation. 

\MainTheorem*

%\begin{theorem}\label{thMain}
%\PName admits a polynomial-time bicriteria $(1+\epsilon, O(\frac{\sqrt{B}}{\epsilon^3}))$-approximation algorithm.% that violates the budget constraint by a factor of $1+\epsilon$.% when for any $v \in V$, $c(v)=1$.
%\end{theorem}

\iffalse
To prove Theorem~\ref{thMain}, we provide some notations and observations. First we need the following claim from Kuo et al.~\cite{kuo2014maximizing}.
\begin{claim}[Claim 3 in Kuo et al.~\cite{kuo2014maximizing}] \label{clClaimKuo}
For any tree $T$ with cost $c(T)$, there always exist $n'=\max(5\lfloor \frac{c(T)}{m}\rfloor, 1)$ out-subtrees $T^i=(V^i, A^i)$ of  $T=(V, E)$, where $c(V^i) \le m+c(r_i)$ and $r_i$ is the root of $T^i$ for all $1 \le i \le n'$, such that $\bigcup_{i=1}^{n'} V_i=V$.
\end{claim}
\fi

For our analysis, we need to decompose an optimal out-tree into a bounded number of out-subtrees of bounded cost as in the following lemma, which is similar to Claim~3 in Kuo et al.~\cite{kuo2014maximizing} on the unrooted problem and undirected graphs.

%\begin{claim} \label{clClaimKuoExtension}
%For any out-tree $T$ rooted at $r$ with $k$ vertices, there always exist $n'=max(5\lfloor \frac{k}{m}\rfloor, 1)$ out-subtrees $T^i=(V^i, A^i)$ of $T=(V, E)$, where $|V^i| \le m$ for all $1 \le i \le n'$ and at most one out-subtree $T^{n'+1}$, where $|V^{n'+1}| < m$ such that $\bigcup_{i=1}^{n'+1} V_i=V$.
%\end{claim}

%Given a budget $B'$ and a root node $v$, let $OPT^{v}_{B'}$ be an optimal solution to an instance $( D=(V, A), p, c, v, B')$ of \textbf{NW-DRBP} and let $OPT = OPT^{r}_{B}$. 

%, let $OPT_B$ be an optimal solution for an instance $I=( D=(V, A), p, c, r, B)$ of \textbf{NW-DRBP} and $OPT_{c(r')+B'}$ be optimal solutions to instances $I=( D=(V, A), p, c, r, B)$ and $J=( D=(V, A), p, c, r', c(r')+B')$ of , respectively. 
%\begin{lemma}\label{lmLemmaKuoExtension}
%$p(OPT) \le 6 \lfloor \sqrt{B} \rfloor p(OPT^{v^*}_{c(v^*)+\sqrt{B}})$, where $v^*=\arg\max_{v \in V} p(OPT^v_{c(v)+\sqrt{B}})$.
%\end{lemma}

%\begin{proof}

%To prove the lemma, we need the following claim.
\begin{lemma} \label{lmClaimKuoExtension}
For any out-tree $\hat{T}=(V, A)$ rooted at $r$ with cost $c(\hat{T})$ and any $m\leq c(\hat{T})$, there exist $N \le 5\lfloor \frac{c(\hat{T})}{m}\rfloor$ out-subtrees $T^i=(V^i, A^i)$ of $\hat{T}$, for $i \in [N]$, where $V^i \subseteq V$, $A^i = (V^i \times V^i) \cap A$, $c(V^i) \le m+c(r_i)$, $r_i$ is the root of $T^i$, and $\bigcup_{i=1}^{N} V^i=V$.%, and at most one out-subtree $T^{N+1}$, where $c(T^{N+1}) < m$ such that $\bigcup_{i=1}^{N+1} V_i=V$.
\end{lemma}
\begin{proof}
An out-subtree $T'$ of $\hat{T}$ rooted at $r'$ is called \textit{feasible} if $c(V(T')\setminus \{r'\}) \le m$; it is called \emph{infeasible} otherwise.
%Recall that the full out-subtree of a vertex $v \in V$, denoted by $\hat{T}_v$, is the out-subtree rooted at $v$ containing the set of all reachable vertices from $r$ through $v$ in $\hat{T}$. Note that $\hat{T}$ itself is the full out-subtree of $r$. 

%We compute $N+1$ out-subtrees $T_1, \dots, T_{N+1}$ of $T$ by using the following 
Let us consider the following procedure called \textbf{Proc}. \textbf{Proc} takes as input an out-tree $T'$, \textbf{Proc}$(T')$, and visits the vertices on $T'$ from the leaves to the root. In this visiting process when \textbf{Proc} encounters a vertex $v$ such that $T'_v$ is the first infeasible full out-subtree, it removes $T'_v$ from $T'$, i.e., $T'=T'\setminus T'_v$. \textbf{Proc} iteratively repeats this process for the new tree $T'$. Finally, \textbf{Proc} returns all infeasible full out-subtrees that have been found in the visit. 

Let $I_1, \dots, I_{s}$ be the set of all infeasible full out-subtrees that have been returned after running \textbf{Proc}$(\hat{T})$ and let $I_{s+1}$ be the possible feasible out-subtree rooted at $r$ that remains after the visit of \textbf{Proc}$(\hat{T})$. We have $\cup_{i\in [s+1]} V(I_i) = V(\hat{T})$  and $V(I_i)\cap V(I_j)=\emptyset$, for $i\neq j$.

For each $i\in [s]$, let us consider the infeasible out-subtree $I_i$, let $v_i$ be the root of $I_i$, and let $I_{u}$ be the full out-subtree of $I_i$ rooted at $u$, for each child $u$ of $v_i$. Each out-subtree $I_i$ is further divided into out-subtrees as follows:
\begin{itemize}
    \item for all children $u$ of $v_i$ such that $c(I_{u})\geq m/2$, we generate an out-subtree $I_{u}$, observe that $c(I_{u}) \leq m + c(u)$ because $I_{u}$ is feasible. If all the children of $v_i$ are in this category, we generate a further out-subtree made of only node $v_i$.
    \item All children $u$ of $v_i$ such that $c(I_{u})< m/2$ are partitioned into groups of cost between $m/2$ and $m$, plus a possible group of cost smaller than $m/2$.
    It is always possible to partition the nodes in this way since $c(I_{u})< m/2$ for all such nodes.
    Then, for each of these groups, we generate an out-subtree by connecting $v_i$ to the roots of the out-subtrees in the group. All the generated out-subtrees have the same root $v_i$ and cost at most $m+c(v_i)$.
\end{itemize}
The generated out-subtrees cover all the nodes in $I_i$. We add $I_{s+1}$ to the set of generated out-subtrees, if it exists.
Let $T^1,\ldots,T^N$ be the set of generated out-subtrees.
Since $I_1, \dots, I_{s+1}$ cover all the nodes of $\hat{T}$, then so do $T^1,\ldots,T^N$. Moreover, each generated out-subtree $T^j$ costs at most $m+c(r_j)$, where $r_j$ is the root of $T^j$.

We now bound the number $N$ of generated out-subtrees.
Given an infeasible out-subtree $I_i$, for some $i\in[s]$, each out-subtree generated from $I_i$ costs at least $m/2$, except for the possible out-subtree made of only the root node of $I_i$ and a possible out-subtree of cost smaller than $m/2$. Note that, by construction, at most one of these two additional out-subtrees can be generated. 
Hence, for each $i\in [s]$, the number $s_i$ of out-subtrees generated from $I_i$ is
    \[
    s_i\le\left\lfloor \frac{c(I_i)}{m/2} \right\rfloor+1 \le 2\left\lfloor \frac{c(I_i)}{m} \right\rfloor+2 \le 4 \left\lfloor \frac{c(I_i)}{m} \right\rfloor.
    \]
Since $I_1, \dots, I_{s+1}$ are disjoint, then the overall number of generated out-subtrees is at most $N\leq 1+\sum_{i\in [s]} s_i\le1+\sum_{i\in[s]} 4 \left\lfloor \frac{c(I_i)}{m} \right\rfloor\le 5 \left\lfloor \frac{c(\hat{T})}{m} \right\rfloor$.
\end{proof}

%Hence by Claim~\ref{lmClaimKuoExtension}, there always exist $N \le 6 \lfloor \sqrt{B} \rfloor$ out-subtrees $T^i=(V^i, A^i)$ of $OPT$, where $V^i \subseteq V$, $A^i=V^i \times V^i \cap A$ and $c(V^i) \le c(u_i)+\sqrt{B}$ for all $1 \le i \le N$ such that $\bigcup_{i=1}^{N} V^i=OPT$. Let $w=\arg\max_{i \in [N]} p(V^i)$. By the definition of $OPT^{v^*}_{c(v^*)+\sqrt{B}}$, we have $p(V^{w}) \le p(OPT^{v^*}_{c(v^*)+\sqrt{B}})$. Hence, by the submodularity of $p$ we have:

%\[
%p(OPT)=p\left(\bigcup_{i=1}^{N} V^i\right) \le \sum_{i=1}^{N}p(V^i) \le 6 \lfloor \sqrt{B} \rfloor p(V^{w}) \le 6 \lfloor \sqrt{B} \rfloor p(OPT^{v^*}_{c(v^*)+\lfloor \sqrt{B} \rfloor}),
%\]
%which completes the proof of the lemma.
%Let $u''=\argmax_{u \in V} c(u)$. Note that the last equation holds as $c(V^i) \le c(u_i)+\sqrt{B}$ for all $1 \le i \le n'$ and $V^i$ is an feasible solution to \textbf{NW-DBP} when the budget should be at most $c(u'')+\lfloor \sqrt{B} \rfloor$, meaning that $p(V^i) \le p(OPT_{c(u')+\lfloor \sqrt{B} \rfloor})$. This completes the proof.
%\end{proof}

Now we are ready to prove Theorem~\ref{thMain}.

\begin{proof}[Proof of Theorem~\ref{thMain}]
\iffalse
Let us first recall some important notations from Algorithm~\ref{algNW-DRBP}. $V_u$ is the set of all vertices that have distance no more than $c(u)+\lfloor \sqrt{B} \rfloor$ from $u$, for any $u \in V$. $I^{RSM}_u$ is an instance of \textbf{RSM} where the elements are $V_u$, the specific element is $u$, the budget is $\lfloor \sqrt{B} \rfloor$, and profits are defined by function  $p(\cdot)$. $S_u$ is a $(1-e^{-1})$-approximate solution to $I^{RSM}_u$, using  \textbf{Greedy}. % the greedy algorithm of Nemhauser et al.~\cite{nemhauser1978analysis}.
$T_u$ is the minimal inclusion-wise out-tree rooted at $u$ spanning all nodes in $S_u$. $z=\arg\max_{u\in V} p(T_u)$. $OPT$ is an optimal solution to \PName, respecting the budget constraint $B$. %Recall that $OPT^v_{c(v)+B'}$ is an optimal solution to an instance $J=( D=(V, A), p, c, v, c(v)+B')$ of \PName and $v^{*}=\arg\max_{v \in V} p(OPT^v_{c(v)+\sqrt{B}})$.
Recall that $P$ is a shortest path from $r$ to $z$, $T= P \cup T_z$ and $E(T):= E(T) \setminus \{ (v,v')\in E(T_z)\setminus E(P)~:~ v'\in V(T_{z})\cap V(P)  \}$. 
\fi

By applying Lemma~\ref{lmClaimKuoExtension} to an optimal solution $T^*$ and by setting $m=\lfloor \sqrt{B} \rfloor$, we obtain $N \le 5 \lfloor \sqrt{B} \rfloor$ out-subtrees $T^i=(V^i, A^i)$ of $T^*$, for $i\in [N]$, where $V^i \subseteq V(T^*)$, $A^i=(V^i \times V^i) \cap A(T^*)$, $c(V^i) \le c(r_i)+\lfloor \sqrt{B} \rfloor$, $r_i$ is the root of $T^i$, and $\bigcup_{i=1}^{N} V^i=T^*$.
Let $p(T') = \max\{ p(T^i ): i\in [N] \}$ and $w$ be the root of $T'$. %By the definition of $OPT^{v^*}_{c(v^*)+\sqrt{B}}$, we have $p(V^{w}) \le p(OPT^{v^*}_{c(v^*)+\sqrt{B}})$. 
The submodularity of $p$ implies $p(T^*) = p\left( \bigcup_{i=1}^N V(T^i)  \right) \leq N p(T'),$
%\begin{align*}%\label{eqPartitionedOPT}
%p(OPT)=p\left(\bigcup_{i=1}^{N} V^i\right) \le \sum_{i=1}^{N}p(V^i) \le 6 \lfloor \sqrt{B} \rfloor p(V^{i^*})
%p(T^*) = p\left( \bigcup_{i=1}^N V(T^i)  \right) \leq N p(T'),
%\le 6 \lfloor \sqrt{B} \rfloor p(OPT^{v^*}_{c(v^*)+\lfloor \sqrt{B} \rfloor}),
%\end{align*}
which implies
%\[
%(1-e^{-1})p(OPT) \le (1-e^{-1}) 6 \lfloor \sqrt{B} \rfloor p(OPT^w_{c(w)+\lfloor \sqrt{B} \rfloor}) \le  6 \lfloor \sqrt{B} \rfloor p(T_z).
%\]
\begin{align}\label{eq:main}
   \!\!\!\!\!\!\!\! p(T) \ge p(T_z) \ge p(S_w)  \ge (1-e^{-1}) p(S^*_w) \ge (1-e^{-1}) p(T') \ge \frac{1-e^{-1}}{N}p(T^*) \ge \frac{1-e^{-1}}{5\lfloor \sqrt{B} \rfloor}p(T^*),
\end{align}
where the first two inequalities hold by the definitions of $z$ and $S_w$ and by the monotonicity of function $p$; The Third inequality holds because $S_w$ is a $(1-e^{-1})$-approximate solution for instance $I^{RSM}_w$;
The fourth inequality holds as (i) $T'$ contains nodes at a distance no more than $c(w)+\lfloor\sqrt{B}\rfloor$ from $w$ and contains at most $1+\lfloor\sqrt{B}\rfloor$ nodes (since the minimum cost of a node is at least 1) and (ii) $p(S^*_w) = \max\{p(S): |S|\leq 1+\lfloor\sqrt{B}\rfloor \text{ and } dist(w,v)\leq c(w)+\lfloor\sqrt{B}\rfloor \text{, for all } v\in S\}$.

%$V(T') \subseteq V_w$, as $V(T')$ contains vertices that have distance no more than $c(w)+\sqrt{B}$ from $w$ and $c(T') \le c(w)+\sqrt{B}$. The last inequality holds by Equation~\eqref{eqPartitionedOPT}.
%\[
%6 \lfloor \sqrt{B} \rfloor p(T) \ge 6 \lfloor \sqrt{B} \rfloor p(T_z) \ge 6 \lfloor \sqrt{B} \rfloor p(OPT^w_{c(w)+\lfloor \sqrt{B} \rfloor}) \ge 6 \lfloor \sqrt{B} \rfloor p(S_w \cup \{w\}) \ge 6 \lfloor \sqrt{B} \rfloor p(S_w) \ge (1-e^{-1})p(OPT).
%\]

Before the trimming process in Lemma~\ref{lmBateniTrimmingProcessGeneral}, the ratio between the prize and the cost of $T$ is at least $\gamma= \frac{1-e^{-1}}{10\sqrt{B} B} p(T^*)$ as $c(T)\le 2B$. After applying the trimming process in Lemma~\ref{lmBateniTrimmingProcessGeneral} (with $h=2$) to $T$, the cost of $T$ is at most $(1+\epsilon)B$ and its prize-to-cost ratio is:
\[
\frac{p(T)}{c(T)}\geq\frac{\epsilon^2 \gamma}{64} = \alpha\frac{\epsilon^2}{\sqrt{B}B} p(T^*),
\]
where $\alpha=\frac{1-e^{-1}}{640}$. As $c(T) \ge \epsilon B/2$, we have $p(T) \geq  \frac{\alpha\epsilon^3}{2\sqrt{B}} p(T^*)$, which concludes the proof.
\end{proof}

\section{The unrooted version of \PName}\label{sec:unrooted}

Here we consider the unrooted version of \PName, denoted by \textbf{DUST}, in which the goal is to find an out-tree $T$ of $D$ such that $c(T)\le B$ and $p(T)$ is maximum. Note that $T$ can be rooted at any vertex. %no distinguished vertex $r$ should be considered as the root of the out-subtree.
By guessing the root of an optimal solution, we can apply the algorithm in the previous section to obtain a bicriteria $(1+\epsilon, O(\frac{\sqrt{B}}{\epsilon^3}))$ approximation. We now show that \textbf{DUST} admits an $O(\sqrt{B})$-approximation algorithm with no budget violations. To do this, we first provide an unrooted version of Lemma~\ref{lmBateniTrimmingProcessGeneral} in which it is not necessary to violate the budget constraint when each vertex costs at most half of the budget. This trimming process follows the same procedure as that of Lemma~\ref{lmBateniTrimmingProcessGeneral}, but we include it for the sake of completeness.% the proof is deferred to the Appendix.

\begin{restatable}{lemma}{TrimmingProcessGeneralUnrooted}\label{lmBateniTrimmingProcessGeneralUnrooted}
Let $T$ be an out-tree with the prize-to-cost ratio $\gamma=\frac{p(T)}{c(T)}$, where $p$ is a monotone submodular function. Suppose $\frac{B}{2}\le c(T) \le hB$, where $h \in (1, n]$ and the cost of each vertex is at most $\frac{B}{2}$. One can find an out-subtree $\hat{T} \subseteq T$ with the prize-to-cost ratio at least $\frac{ \gamma}{32h+8}$ such that $B/4 \le c(\hat{T}) \le B$.
\end{restatable}

%We first generalize the trimming process of Bateni et al.~\cite{bateni2018improved} for the case when the graph is directed, the prize function is monotone submodular and it is not necessary to have a distinguished vertex $r$ in the resulting out-subtree. This trimming process follows the same process as that of Lemma~\ref{lmBateniTrimmingProcessGeneral}, but the proof is included for the sake of completeness. 

%\TrimmingProcessGeneralUnrooted*
\begin{proof}
%We make $T$ rooted at an arbitrary vertex $r$ 
In the initial step, we remove a strict out-subtree $T'$ of $T$ if (i) the prize-to-cost ratio of $T\setminus T'$ is at least $\gamma$, and (ii) $c(T\setminus T') \ge \frac{B}{4}$. This process is performed iteratively, until no such out-subtree exists. Let $T_{-}$ be the remaining out-subtree after applying this iterative process on $T$.

If $c(T_{-}) \le B$, the desired out-subtree is obtained and we are done. Suppose it is not the case. A full out-subtree $T'$ is called \emph{rich} if $c(T') \ge \frac{B}{4}$ and the prize-to-cost ratio of $T'$ and all its strict out-subtrees are at least $\gamma$. As in Lemma~\ref{lmBateniTrimmingProcessGeneral}, we claim that the lemma follows from the existence of a rich out-subtree.

\begin{claim}\label{clBateniRichGeneralUnrooted}
Given a rich out-subtree $T'$, the desired out-subtree $\hat{T}$ can be found.
\end{claim}

\begin{proof}
Let $T''$ be the lowest rich out-subtree of $T'$ such that the strict out-subtrees of $T''$ are not rich, i.e., $c(T'') \ge \frac{B}{4}$ while the cost of strict out-subtrees of $T''$ (if any exist) is less than $\frac{B}{4}$. Let $C$ be the total cost of the immediate out-subtrees of $T''$. We distinguish between two cases:

\begin{enumerate}
    \item If $C < \frac{B}{4}$, then $c(T'') \le \frac{3B}{4}$ as the root of $T''$ costs at most $\frac{B}{2}$. Since $T''$ has the prize-to-cost ratio at least $\gamma$ and cost at least $\frac{B}{4}$ (as it is rich), $\hat{T}=T''$ is the desired out-subtree.
    
    \item If $C \ge \frac{B}{4}$, we first group the immediate out-subtrees of $T''$ into $M$ groups $S_1, \dots, S_M$ in such a way that for each $i\in [M-1]$ the total cost of immediate out-subtrees in $S_i$ is at least $\frac{B}{4}$, and for each $i\in[M]$ it is at most $\frac{B}{2}$. As $c(T_{-}) \le hB$, we have 
    \[
    M \le \left\lceil \frac{hB}{B/4} \right\rceil= \left\lceil 4h \right\rceil \le 4h+1.
    \]
   For each $i\in [M]$, let $S'_i=S_i \cup \{r''\}$, where $r''$ is the root of $T''$. Let $z=\arg\max_{i \in [M]} p(S'_i)$. Hence by the submodularity and monotonicity of $p$, we have 
   
    \[
    p(S'_{z}) \ge \frac{\sum_{i=1}^{M} p(S'_i)}{4h+1} \ge \frac{p(S'_1) + \sum_{i=2}^{M}p(S_i)}{4h+1}\ge \frac{p(S'_1 \cup\bigcup_{i=2}^{M} S_i)}{4h+1} = \frac{p(T'')}{4h+1} \ge \frac{\gamma}{16h+4}B,
    \]
   
   where the last inequality holds as $p(T'') \ge \gamma\frac{B}{4}$ (since $T''$ is rich). 
   
   In case $z=M$ and $c(S'_M) <\frac{B}{4}$, we select a subset of immediate out-subtrees from $\bigcup_{i=1}^{M-1}S_i$ with the total cost of at least $\frac{B}{4}$ and at most $\frac{B}{2}-c(S_M)$, and add it to $S_z$.

   Let $\hat{T}$ be the union of $r''$, the edges from $r''$ to the roots of the out-subtrees in $S_{z}$, and $S_{z}$. The cost of $\hat{T}$ is at most $B$. Hence the prize-to-cost ratio of $\hat{T}$ is at least $\frac{\gamma}{16h+4} \ge \frac{\gamma}{32h+8}$.\qedhere
\end{enumerate}
\end{proof}

It only remains to consider the case when there is no rich out-subtree. Since $T_{-}$ is not rich and $c(T_{-}) \ge \frac{B}{4}$, the ratio of at least one of the strict out-subtrees of $T_{-}$ is less than $\gamma$. Now we find an out-subtree $T'$ with ratio less than $\gamma$ such that the ratio of all of its strict out-subtrees (if any exist) is at least $\gamma$. Since the ratio of $T'$ is less than $\gamma$ and $T'$ is not removed in the initial process, $c(T_-\setminus T') < \frac{B}{4}$ (this can be shown by the same argument as that of Claim~\ref{clLowCost}). As $c(T_{-})> B$ and the cost of the root of $T'$ is at most $\frac{B}{2}$, the total cost of the immediate out-subtrees of $T'$ is at least $\frac{B}{4}$. Also, the cost of an immediate out-subtree of $T'$ is less than $\frac{B}{4}$, otherwise we have a rich out-subtree. As the ratio and cost of $T_{-}$ are at least $\gamma$ and $\frac{B}{4}$, respectively, then $p(T_{-}) \ge \frac{\gamma}{4}B$. We distinguish between two cases.

\begin{enumerate}
    \item  If $p(T') \ge \frac{\gamma}{8}B$, by the similar reasoning as above, we partition the immediate out-subtrees of $T'$ into $M$ groups $S_1, \dots, S_M$ in such a way that for each $i\in [M-1]$ the total cost of immediate out-subtrees in $S_i$ is at least $\frac{B}{4}$, and for each $i\in[M]$ it is at most $\frac{B}{2}$. For each $i\in [M]$, let $S'_i=S_i \cup \{r'\}$ where $r'$ is the root of $T'$. Let $z=\arg\max_{i \in [M]} p(S'_i)$. As $M \le 4h+1$, by the submodularity and monotonicity of $p$ we have: 
     \[
    p(S'_{z}) \ge \frac{\sum_{i=1}^{M} p(S'_i)}{4h+1} \ge \frac{p(S'_1) + \sum_{i=2}^{M}p(S_i)}{4h+1}\ge \frac{p(S'_1 \cup\bigcup_{i=2}^{M} S_i)}{4h+1} =\frac{p(T')}{4h+1} \ge \frac{\gamma}{32h+8}B,
    \]
    where the last inequality holds as $p(T') \ge \frac{\gamma}{8} B$. 
    
    Note that in case $z=M$ and $c(S'_M) <\frac{B}{4}$, we select a subset of immediate out-subtrees from $\bigcup_{i=1}^{M-1}S_i$ with the total cost of at least $\frac{B}{4}$ and at most $\frac{B}{2}-c(S_M)$, and add it to $S_z$. 
    
    Let $\hat{T}$ be the union of $r'$, the edges from $r'$ to the roots of the out-subtrees in $S_{z}$ and $S_{z}$. The cost of $\hat{T}$ is at most $B$ and its prize-to-cost ratio is at least $\frac{\gamma}{32h+8}$.
    
    \item If $p(T') < \frac{\gamma}{8}B$, we proceed as follows. Consider the out-subtree $T''=T_{-}\setminus T'$. %, which is rooted at an arbitrary vertex $r$
    Recall that by the above discussion we have $c(T'') <\frac{B}{4}$. Thus we find a subset $S$ of the immediate out-subtrees of $T'$ with cost between $\frac{B}{4}-c(T'')\le  c(S) \le \frac{B}{2}-c(T'')$. Note that such set $S$ can be found as each immediate out-subtree of $T'$ costs less than $\frac{B}{4}$ (otherwise we have a rich subtree) and $c(T'') > \frac{3 B}{4}$ (as $c(T_{-}) > B$ and $c(T'') <\frac{B}{4}$). Then let $\hat{T}$ be the union of $T''$, the edge from $T''$ to $r'$ in $T_{-}$, $S$, and the edges from $r'$ to the roots of the out-subtrees in $S$, where $r'$ is the root of $T'$. We now bound the prize-to-cost ratio of $\hat{T}$. Recall that $T''=T_{-}\setminus T'$. First note that by the submodularity' properties $p(T'')+p(T') \ge f(T_{-})$. Thus by the case assumption and monotonicity, we have $f(\hat{T}) \ge p(T'') \ge \frac{\gamma}{8}B$. Since $\frac{B}{4}B-c(T'') \le c(S) \le \frac{B}{2}-c(T'')$ and $c(r') \le \frac{B}{2}$, $c(\hat{T}) \le B$. Therefore, the prize-to-cost ratio of $\hat{T}$ is at least $\frac{\gamma}{8} \ge \frac{\gamma}{32h+8}$.
\end{enumerate}
The proof is complete.
\end{proof}

\begin{theorem}\label{thMainUndirectedEW-DRBP}
\textbf{DUST} admits a polynomial-time $O(\sqrt{B})$-approximation algorithm.
\end{theorem}

\begin{proof}
%Note that for \textbf{DUST}, we only need to run Lines 1-8 of Algorithm~\ref{algNW-DRBP} to find an out-tree $T$ with cost $c(T)\le 2B$ and prize $p(T) \ge O(\frac{1}{\sqrt{B}})opt$, where $opt$ is the optimum solution to \textbf{DUST}. This can be proved by the same argument as that of Theorem~\ref{thMain}. 

%Now if the cost of each vertex is at most $B/2$ in $\hat{T}$, from $T$, by Lemma~\ref{lmBateniTrimmingProcessGeneralUnrooted}, an out-subtree $\hat{T}$ with cost $c(\hat{T})\le B$ and prize $p(\hat{T}) \ge O(\frac{1}{\sqrt{B}})opt$ can be obtained. In the rest of the proof, we show that the desired out-subtree can be achieved even when there are some vertex costing more than $B/2$ in the input graph.
We follow arguments similar to those in Theorem 4 from Bateni et al.~\cite{bateni2018improved}, but for the sake of completeness the proof is provided here.

An out-tree is called \textit{flat} if each vertex of the out-tree costs no more than $\frac{B}{2}$. Let $x$ be a vertex of an out-tree with the largest cost. An out-tree is called \textit{saddled} if $c(x) > \frac{B}{2}$ and the cost of every other vertex of the out-tree is no more than $\frac{B-c(x)}{2}$. Let $T^*_f$ (resp. $T^*_s$) be the optimal flat (resp. saddled) out-tree, i.e, a flat (resp. saddled) out-tree with cost at most $B$ maximizing the prize. 
We first show that given an optimal solution $T^*$ to \textbf{DUST}, then either $p(T^*_f) \ge \frac{p(T^*)}{2}$ or $p(T^*_s) \ge \frac{p(T^*)}{2}$. 

\begin{restatable}{claim}{FlatSaddeledBateniGeneralize}\label{clFlatSaddeledBateniGeneralize}
Either $p(T^*_f) \ge \frac{p(T^*)}{2}$ or $p(T^*_s) \ge \frac{p(T^*)}{2}$, where $T^*$ is an optimal solution to \textbf{DUST}.
\end{restatable}

\begin{proof}
If $T^*$ has only one vertex, then it is either flat or saddled and we are done. If $T^*$ has more than one vertex and it is neither flat nor saddled, then we proceed as follows. Let $x$ and $y$ be two vertices in $T^*$ with the maximum cost and the second maximum cost, respectively. Since $T^*$ is not flat then $c(x) > \frac{B}{2}$ and $c(y) \le \frac{B}{2}$. Also as $T^*$ is not saddled, $c(y) > \frac{B-c(x)}{2}$, and, since the cost of $T^*$ is at most $B$, $y$ is the only node with a cost higher than $\frac{B-c(x)}{2}$. By removing the edge $e$ adjacent to $y$ on the path between $x$ and $y$, we can partition $T^*$ into two out-subtrees $T^*_x$ and $T^*_y$ that contain $x$ and $y$, respectively. Clearly, each vertex in $T^*_y$ costs no more than $\frac{B}{2}$, then $T^*_y$ is flat. Also, each vertex in $T^*_x$ except $x$ costs at most $\frac{B-c(x)}{2}$, implying that $T^*_x$ is saddled. By the submodularity of $p$, $p(T^*_x)+p(T^*_y) \ge p(T^*)$, meaning that one of $T^*_x$ and $T^*_y$ has at least half of the optimum prize $p(T^*)$, which concludes the claim.
\end{proof}

%Next, we show that by restricting the algorithm to these two classes the desired solution can be obtained.
Now we restrict Algorithm~\ref{algNW-DRBP} to only flat and saddled out-trees. Indeed, we can reduce the case of saddled out-trees to flat out-trees as follows. We first find a vertex $x$ with the maximum cost. We then set the cost of $x$ to zero and define a new budget $B'=B-c(x)$. Note that the cost of any other vertex in the optimal saddled out-tree $T^*_s$ is at most half of the remaining budget. This means that we only need to find an approximation solution when restricted to flat out-trees.

%By using Lemma~\ref{lmBateniTrimmingProcessGeneralUnrooted}, the desired out-subtree $\hat{T}$ can be found. 
Since for the new instance no other vertex except $x$ with cost more than $\frac{B}{2}$ can be contained in the final solution, we remove all vertices with cost more than $\frac{B}{2}$ and run Lines 1-8 of Algorithm~\ref{algNW-DRBP} on the new resulting graph to achieve an out-tree $T$ with cost $c(T)\le 2B$ (as $c(T \setminus \{z\})=B$ and $c(z) \le B$) and prize $p(T) \ge \frac{1-e^{-1}}{5\sqrt{B}}p(T^*_f) \ge \frac{1-e^{-1}}{10\sqrt{B}}p(T^*)$. So, the prize-to-cost ratio of $T$ is $\gamma\geq\frac{p(T)}{2B}$.
%(this can be proved by the same argument as that of Theorem~\ref{thMain} and Claim~\ref{clFlatSaddeledBateniGeneralize}).
As $T$ is flat, we can apply Lemma~\ref{lmBateniTrimmingProcessGeneralUnrooted} to achieve an out-subtree $\hat{T}$ of $T$ with the cost $B/4 \le c(\hat{T}) \le B$ and the prize-to-cost ratio $\frac{p(\hat T)}{c(\hat T)}\ge\frac{\gamma}{32 h+8}=\frac{\gamma}{72}$ as $h \le 2$. This implies that 
\[
p(\hat{T})\ge\frac{\gamma}{72} \cdot \frac{B}{4}=\frac{\gamma}{288}B\geq\frac{p(T)}{576}\geq \frac{1-e^{-1}}{5760\sqrt{B}}p(T^*).\qedhere
\]
%and prize $p(\hat{T}) \ge \frac{p(T)}{512}=O(1/\sqrt{B})p(T^*)$ (since $h=2$, $c(T)\leq 2B$, and $c(\hat{T})\geq B/4$).
\end{proof}

\section{Submodular Tree Orienteering}\label{sec:STO}

Recently, Ghuge and Nagarajan~\cite{ghuge2020quasi} studied the \emph{Submodular Tree Orienteering} problem (\textbf{STO}), which is similar to \PName with the only difference that the costs are associated to the edges of a directed graph instead of the nodes. In particular, in \textbf{STO}, we are given a directed graph $D=(V, A)$, a vertex $r \in V$, a budget $B$, a monotone submodular function $p: 2^V \rightarrow \mathbb{R}^+$, and a cost $c:A \rightarrow \mathbb{Z}^+$, and the goal is to find an out-subtree $T$ of $D$ rooted at $r$  such that $ \sum_{e\in A(T)}c(e) \le B$ and $p(T)=p(V(T))$ is maximum. Ghuge and Nagarajan~\cite{ghuge2020quasi} proposed an $O\left(\frac{\log k}{\log \log k}\right)$-approximation
algorithm for \textbf{STO} that runs in $(n \log B)^{O(\log^{1+\epsilon} k)}$ time for any constant $\epsilon > 0$, where $k \le |V|$ is the number of vertices in an optimal solution.
\iffalse
\begin{restatable}{theorem}{QuasiPolynomialTheorem}[Ghuge and Nagarajan~\cite{ghuge2020quasi}, Theorem 1]\label{thGhugeEW-RDBPQuasiAlgorithm}
There is an $O(\frac{\log k}{\log \log k})$-approximation
algorithm for \textbf{DTO} that runs in
$(n \log B)^{O(\log^{1+\epsilon} k)}$ time for any constant $\epsilon > 0$, where $k \le n=|V|$ is the number of vertices in an optimal solution.
\end{restatable}
\fi
%provided a quasi-polynomial time $O(\frac{\log n'}{\log \log n'})$-approximation algorithm for \textbf{STO}.

Here we first show that \PName can be reduced to \textbf{STO}, preserving the approximation factor, by assigning the cost of each node $v$ to all edges entering $v$. 

%It follows that \PName admits a quasi-polynomial time $O(\frac{\log n'}{\log \log n'})$-approximation.
\begin{restatable}{theorem}{GhugeNagarajanNWDRBP}\label{thGhugeNagarajanNW-DRBP}
There is an $O(\frac{\log k}{\log \log k})$-approximation
algorithm for \PName that runs in
$(n \log B)^{O(\log^{1+\epsilon} k)}$ time for any constant $\epsilon > 0$, where $k \le n=|V|$ is the number of vertices in an optimal solution.
\end{restatable}
\begin{proof}
To prove the theorem, we show that one can transform an instance $J=\langle D'= (V', A'), p', c', r', B' \rangle$ of \PName to an instance $I=\langle D= (V, A), p, c, r, B \rangle$ of \textbf{STO} as follows. We set $V=V'$, $r=r'$, $A=A'\setminus \{(v, r')| (v, r') \in A'\}$, $B=B'-c'(r')$ and for any subset $S \subseteq V$, $p(S)=p'(S)$. For any $e=(i,j) \in A$ in $I$, we set $c(e)=c'(j)$. The theorem follows by observing that any out-subtree $T$ of $D$ is an out-subtree for $D'$, $c'(T)=\sum_{v \in V(T)} c'(v)=\sum_{e=(u,v) \in A(T)}c(e) + c'(r')=c(T) +c'(r')$, and $p'(T)=p(T)$.
%
%Consider a solution $T$ to \textbf{STO} on $I$. Now we obtain a solution $T'=(V'(T'), A'(T'))$ to \PName from $T=(V(T), E(T))$ as follows. We set $V'(T')=V(T)$ and $A'(T')=A(T)$. We conclude the theorem by noticing that $c'(T')=\sum_{v \in V(T')} c'(v)=\sum_{e=(u,v) \in A(T)}c(e) + c'(r')$ and $p'(T')=p(T)$.
\end{proof}

Moreover, we show that Algorithm~\ref{algNW-DRBP} can be used to approximate \textbf{STO}. To our knowledge, this is the first polynomial-time bicriteria approximation algorithm for \textbf{STO}.

\begin{theorem}\label{thSTOAlgorithm1}
There exists a polynomial-time bicriteria $(1+\epsilon, O(\frac{\sqrt{B}}{\epsilon^3}))$-approximation algorithm for \textbf{STO}, for any $\epsilon \in (0,1]$.
\end{theorem}
\begin{proof}
We first transform an instance $I_S=\langle D_S= (V_S, A_S), p_S, c_S, r, B \rangle$ of \textbf{STO} to an instance $I_D=\langle D_D= (V_D, A_D), p_D, c_D, r, B \rangle$ of \PName, where  $V_D = V_S \cup V_{A}$, $V_{A} = \{v_e : e \in A_S\}$, $A_D = \{(i, v_e), (v_e, j):  e=(i,j)\in A_S\}$, $p_D(S)=p_S(S\cap V_S)$, for each $S\subseteq V_D$, $c_D(v)=0$ for each $v\in V_S$, and $c_D(v_e)=c_S(e)$ for each $v_e\in V_A$.

%\begin{itemize}
 %   \item $V' = V \cup V_{A}$ with $V_{A} = \{v_e : e \in A\}$. 
  %  \item $A' = \{(i, v_e), (v_e, j):  e=(i,j)\in A\}$. 
%\end{itemize}$J=\langle D'= (V', A'), p', c', r', B' \rangle$ of \PName.
%We set $r'=r$ and $B'=B$. Each subset $S \subseteq V'$ is assigned with prize $p'(S)=p(S\setminus V_A)$. Each vertex $v \in V$ (resp. $v_e \in V_{A}$) is assigned with a cost $c'(v)=0$ (resp. $c'(v_e)=c(e)$). %Consider a solution $T'=(V'(T'), A'(T'))$ to \textbf{NW-RDBP} on $J$. %Note that in $T'$ if there is a leaf $l \in V_A$, the next node $v \in e=(l, v)$ can be added to the solution $T'$ as cost $c'(v)=0$ and prize $p'(V'(T') \cup \{l\})\ge p'(V'(T'))$ (by the monotonicity of $p'$). From $T'$, we obtain $T=(V(T), A(T))$ as follows. We set $V(T) = V'(T') \cap V$ and $A(T) = \{ e=(i,j):  (i, v_e), (v_e, j)\in A'(T')\}$. For any $e \in A(T)$, we set $c(e)=c'(v_e)$. The proof follows from $c(T)=\sum_{e=(u,v) \in A(T)}c(e)=\sum_{v_e \in V'(T') \cap V_A} c'(v_e)=c'(T')$ and $p(T)=p'(T')$.

Let $T^*_S$ be an optimal solution for $I_S$ and let $T^*_D$ be the out-subtree of $D'$ corresponding to $T^*_S$ (i.e. $V(T^*_D) = V(T^*_S) \cup \{ v_e : e \in A(T^*_S) \}$, $A(T^*_D) = \{ (i, v_e), (v_e, j):  e=(i,j)\in A(T^*_S) \}$). We observe that $p_D(T^*_D) = p_S(T^*_S)$, $c_D(T^*_D) = c_S(T^*_S)$, and $T^*_D$ is an optimal solution for $I_D$, since if there exists an out-subtree $\bar{T}_D$ of $D_D$ with $p_D(\bar{T}_D) > p_D(T^*_D)$, then we can construct an out-subtree $\bar{T}_S=(V(\bar{T}_D)\cap V_S, \{e\in A_S:v_e\in V(\bar{T}_D)\cap V_A\})$ of $D_S$ such that $p_S(\bar{T}_S)  >  p_S(T^*_S)$. 
%Moreover, $p_D(T^*_D) = p_S(T^*_S)$.

We decompose $T^*_D$ as in Lemma~\ref{lmClaimKuoExtension},\footnote{Note that the Lemma~\ref{lmBateniTrimmingProcessGeneral} and~\ref{lmClaimKuoExtension} hold even if node costs are allowed to be equal to 0.} with $m=\lfloor \sqrt{B}\rfloor$; let $T'_D$ be the out-subtree that maximizes the prize among those returned by the lemma, and let $w$ be the root of $T'_D$. We have that $p_D(T'_D)\ge \frac{1}{5\lfloor \sqrt{B} \rfloor}p_D(T^*_D)$ and $c(T'_D) \leq c(w)+ \lfloor\sqrt{B}\rfloor$.
It follows that the distance from $w$ to any other node in $T'_D$ is at most $c(w)+\lfloor\sqrt{B}\rfloor$.

We now show that $|V(T'_D)\cap V_S|\leq  \lfloor\sqrt{B}\rfloor +1$. Since the cost of nodes in $V_S$ is equal to 0, then $c((V(T'_D)\cap V_A)\setminus\{w\}) = c(V(T'_D)\setminus\{w\}) \leq \lfloor\sqrt{B}\rfloor$. Therefore, as the cost of each edge in $A_S$ is at least 1, $|(V(T'_D)\cap V_A)\setminus\{w\} |\leq \lfloor\sqrt{B}\rfloor$.
For every node in $(V(T'_D)\cap V_S) \setminus \{w\}$, there exists a distinct node in $V(T'_D)\cap V_A$, which means that $|(V(T'_D)\cap V_S) \setminus \{w\}|=|V(T'_D)\cap V_A|$. If $w\in V_S$, then $|V(T'_D)\cap V_A| = |(V(T'_D)\cap V_A)\setminus\{w\} |\leq \lfloor\sqrt{B}\rfloor$. If $w\in V_A$, then $|(V(T'_D)\cap V_A) |\leq \lfloor\sqrt{B}\rfloor +1$. In both cases $|V(T'_D)\cap V_S|\leq  \lfloor\sqrt{B}\rfloor +1$.

%We can assume w.l.o.g. that $w\in V_S$ as otherwise we can add the node in $V_S$ and the edge entering to the root of $T'_D$ without increasing its cost and possibly increasing its prize. 
%It follows that $c(w)=0$, $c(T'_D) \leq \lfloor\sqrt{B}\rfloor$, which implies that the distance from $w$ to any other node in $T'_D$ is at most $\lfloor\sqrt{B}\rfloor$. Moreover, we have that $|V(T'_D)\cap V_S|\leq  \lfloor\sqrt{B}\rfloor +1$ since for every node in $V(T'_D)\cap (V_S \setminus \{w\})$ there exist a distinct node in $V(T'_D)\cap V_A$ with a cost at least equal to 1.

Let $T_D$ be the output of lines 1-11 of Algorithm~\ref{algNW-DRBP} for instance $I_D$. 
We have that 
\begin{align*}
    p_D(T_D) \ge (1-e^{-1}) p_D(S^*_w) \ge (1-e^{-1}) p_D(T'_D) \ge \frac{1-e^{-1}}{5\lfloor \sqrt{B} \rfloor}p_D(T^*_D) = \frac{1-e^{-1}}{5\lfloor \sqrt{B} \rfloor}p_S(T^*_S),
\end{align*}
where the second inequality is due to the fact that (i) $T'_D$ contains nodes at a distance no more than $c(w)+\lfloor\sqrt{B}\rfloor$ from $w$ and contains at most $\lfloor\sqrt{B}\rfloor +1$ nodes in $V_S$, and (ii) $p_D(S) = p_D(S\cap V_S)$, for each $S\subseteq V_D$, and therefore $p_D(S^*_w)  = \max\{p_D(S): |S\cap V_S|\leq \lfloor\sqrt{B}\rfloor +1\text{ and } dist(w,v)\leq c(w)+\lfloor\sqrt{B}\rfloor \text{, for all } v\in S\}$. The other inequalities are analogous to those in~\eqref{eq:main}.

The cost of $T_D$ is at most $2B$, as in Theorem~\ref{thMain} we can trim $T_D$ to reduce its cost to $(1+\epsilon)B$ and maintaining a prize of $p_D(T_D)=\frac{\alpha\epsilon^2}{\sqrt{B}}p(T^*_S)$, for some constant $\alpha$ and any arbitrary $\epsilon>0$.

Let us consider the out-subtree $T_S$ of $D_S$ corresponding to $T_D$,  $T_S=(V(T_D)\cap V_S, \{e\in A_S:v_e\in V(T_D)\cap V_A\})$, then $p_S(T_S)=p_D(T_D)$ and $c_S(T_S)=c_D(T_D)$, which concludes the proof.
\end{proof}
\section{Further Results on Some Variants of \PName}\label{sec:variants}
In this section, we provide approximation results on some variants of \PName. All the details are given in Appendix~\ref{apx:variants}. 

\textbf{Additive prize function and Directed Tree Orienteering (\textbf{DTO})}.
We consider the special case of \PName in which the prize function is additive, i.e., for any $S \subseteq V$, $p(S)=\sum_{v \in S} p(\{v\})$, called \textbf{DRAT}.
We show that there exists a polynomial-time bicriteria $(1+\epsilon, O(\sqrt{B}/\epsilon^2))$-approximation algorithm for \textbf{DRAT} (Theorem~\ref{thNWDRBAlgorithm1}).
%Moreover, in the particular case in which each vertex of the graph has at most one incoming edge, there exists a pseudo-polynomial-time algorithm for \textbf{DRAT} whose running time is $O(n^2p_{max})$, where $p_{max}=\max_{v \in V} p(v)$ (Theorem~\ref{thNW-DRBseudo-Polynomial}), and an FPTAS that runs in $O(\frac{n^3}{\epsilon})$, where $\epsilon >0$ (Theorem~\ref{thNW-DRBFPTAS}). 
By using the reduction in Theorem~\ref{thSTOAlgorithm1}, it follows that this result also holds for \textbf{DTO}, which is the special case of \textbf{STO} in which the prize function is additive~\cite{ghuge2020quasi}.
%\textbf{DRAT} is fixed parameter tractable when parametrized by bi-rank-width, directed clique-width, and directed NLC-width, where the cost of each node is $1$ (Theorem~\ref{thCliqueWidth2}).

%\textbf{Directed Tree Orienteering (\textbf{DTO})}. This is the same problem as \textbf{STO} instead the prize function is an additive function. By using Theorems~\ref{thSTOAlgorithm1} and~\ref{thNWDRBAlgorithm1}, we have the following. There exists a polynomial-time bicriteria $(1+\epsilon, O(\frac{\sqrt{B}}{\epsilon^2}))$-approximation algorithm for \textbf{DTO}.

\textbf{Undirected graphs}.
All our results hold also in the case in which the input graph is undirected and the output graph is a tree.
%\textbf{MCSB} is the undirected variant of \PName, that is it takes an undirected graph as input and outputs a tree.
%In \textbf{MCSB}, we are given an undirected graph $G=(V, E)$, a monotone submodular function $f: 2^V \rightarrow \mathbb{R}^+$, a cost function $c: V \rightarrow \mathbb{Z}^+$ and a budget $B$, and the goal is to find a tree $T$ s.t. $c(T)=\sum_{v \in V(T)}c(v) \le B$ and $f(T)$ is maximum.
In particular, our $O(\sqrt{B})$-approximation algorithm for the unrooted case improves over the factors $O((\Delta+1) \sqrt{B})$~\cite{kuo2014maximizing} and $\min\{1/((1-1/e)(1/R-1/B)), B\}$~\cite{hochbaum2020approximation}, where $R$ is the radius of the input graph $G$.
For the case in which the prize function is additive, we show that there exists a polynomial-time bicriteria approximation algorithm whose approximation factor only depends on the the maximum degree $\Delta$ of the given graph. In particular, it guarantees a bicriteria $(1+\epsilon, 16 \Delta/\epsilon^2)$-approximation (Theorem~\ref{thRBPDegree}).%whose approximation only depends on $\Delta$. In particular, it guarantees a bicriteria  $(1+\epsilon, 16 \Delta/\epsilon^2)$-approximation (Theorem~\ref{thRBPDegree}).

%\textbf{Node-Weighted Rooted Budgeted Problem (\textbf{NW-RBP})}. In \textbf{NW-RBP}, we are given an undirected graph $G=(V,E)$, a cost function $c:V \rightarrow \mathbb{R}^+$ and a prize function $p:V \rightarrow \mathbb{R}^+$, a budget $B$ and a vertex $r$, and the goal is to find a tree $T=(V(T), E(T))$ such that $c(V(T)) \le B$, $r \in V(T)$ and $p(T)$ is maximum. We next provide the first approximation algorithm for \textbf{NW-RBP} that only depends on the maximum degree. We show that \textbf{NW-RBP} admits a polynomial-time bicriteria $(1+\epsilon, \frac{16 \Delta}{\epsilon^2})$-approximation algorithm, where $\Delta$ is the maximum degree in the given graph (Theorem~\ref{thRBPDegree}).

\textbf{Quota problem}. We consider the problem in which we are given an undirected graph $G=(V,E)$, a cost function $c:V \rightarrow \mathbb{R}^+$, a prize function $p:V \rightarrow \mathbb{R}^+$, a quota $Q\in \mathbb{R}^+$, and a vertex $r$, and the goal is to find a tree $T$ such that $p(T) \ge Q$, $r \in V(T)$ and $c(T)$ is minimum. We prove that this problem admits a $2\Delta$-approximation algorithm (Theorem~\ref{thNW-RQPDegree}).%, where $\Delta$ is the maximum degree of a node in the given graph (Theorem~\ref{thNW-RQPDegree}).

\textbf{Maximum Weighted Budgeted Connected Set Cover (\textbf{MWBCSC})}. Let $X$ be a set of elements, $\mathcal{S} \subseteq 2^X$ be a collection of sets, $p: X \rightarrow \mathbb{R}^+$ be a prize function, $c: \mathcal{S} \rightarrow \mathcal{R}^+$ be a cost function, $G_{\mathcal{S}}$ be a graph on vertex set $\mathcal{S}$, and $B$ be a budget. In \textbf{MWBCSC}, the goal is to find a subcollection $\mathcal{S}'\subseteq \mathcal{S}$ such that $c(\mathcal{S}')=\sum_{S \in \mathcal{S}'}c(S) \le B$, the subgraph induced by $\mathcal{S}'$ is connected and $p(\mathcal{S}')=\sum_{x \in X_{\mathcal{S}'}} p(x)$ is maximum, where $X_{\mathcal{S}'}=\bigcup_{S \in \mathcal{S}'} S$.
We show that \textbf{MWBCSC} admits a polynomial-time $\alpha f$-approximation algorithm, where $f$ is the maximum frequency of an element and $\alpha$ is the performance ratio of an algorithm for the unrooted version of \textbf{MCSB} with additive prize function (Theorem~\ref{thMWBCSCGeneral}).
%When each set costs $1$, \textbf{MWBCSC} admits a polynomial-time $2f$-approximation algorithm, where $f$ is the maximum frequency of an element (Theorem~\ref{thMWBCSCwithCost1}). Let $\mathcal{A}$ be an $\alpha$-approximation algorithm for \textbf{NW-BP}. \textbf{MWBCSC} admits a polynomial-time $(h+1)\alpha$-approximation algorithm, where $h=\max_{\substack{S, S' \in \mathcal{S}\\S \cap S' \ne \emptyset}} dist(S,S')-(c(S)+c(S))$ such that $c(S) \ge 1, \forall{S \in \mathcal{S}}$ (Theorem~\ref{thMWBCSCDistance}).
Moreover, one can have a polynomial-time $O(\log{n})$-approximation algorithm for \textbf{MWBCSC} under the assumption that if two sets have an element in common, then they are adjacent in $G_{\mathcal{S}}$ (Corollary~\ref{assumAdjacencyCMCP}). This last result is an improvement over the factor $2(\Delta+1)\alpha/(1-e^{-1})$ by Ran et al.~\cite{ran2016approximation}.

\textbf{Budgeted Sensor Cover Problem (\textbf{BSCP})}. In \textbf{BSCP}, we are given a set $\mathcal{S}$ of sensors, a set $\mathcal{P}$ of target points in a metric space, a sensing range $R_s$, a communication range $R_c$, and a budget $B$. A target point is covered by a sensor if it is within distance $R_s$ from it. Two sensors are connected if they are at a distance at most $R_c$.
The goal is to find a subset $\mathcal{S}' \subseteq \mathcal{S}$ such that $|\mathcal{S}'| \le B$, the number of covered target points by $\mathcal{S}'$ is maximized and $\mathcal{S}'$ induces a connected subgraph. 
%This problem usually is investigated under the assumption that $R_c\le 2R_s$, where $R_s$ and $R_c$ are the sensing range and communication range of sensors, respectively.
We give a $2f$-approximation algorithm for \textbf{BSCP} (Theorem~\ref{thMWBCSCwithCost1}), where $f$ is the maximum number of sensors that cover a target point. We also show that,  under the assumption that $R_s\le R_c/2$,  \textbf{BSCP} admits a polynomial-time $8/(1-e^{-1})$-approximation algorithm (Theorem~\ref{thBSCPImproved}), which improves the factors $8(\lceil 2\sqrt{2} C\rceil+1)^2/(1-1/e)$~\cite{huang2020approximation} and $8(\lceil 4C/\sqrt{3}\rceil+1)^2/(1-1/e)$~\cite{yu2019connectivity}, where $C=R_s/R_c$.
%
%to the literature~\cite{gao2018algorithm, huang2020approximation, kuo2014maximizing, ran2016approximation, xu2021throughput, yu2019connectivity}. 
Note that Huang et al.~\cite{huang2020approximation} do not assume that $R_s/R_c\le R_c/2$, however, our technique improves over their result if $R_s\le R_c/2$.

%over the factors $O(\log{n})$~\cite{gao2018algorithm, ran2016approximation}, $O(\sqrt{B})$~\cite{gao2018algorithm, kuo2014maximizing, xu2021throughput} and $8(\lceil 2\sqrt{2} C\rceil+1)^2/(1-1/e)$~\cite{huang2020approximation} and $8(\lceil 4C/\sqrt{3}\rceil+1)^2/(1-1/e)$~\cite{yu2019connectivity} where $C=R_s/R_c\le 1$.

%\bibliographystyle{unsrt}
\bibliographystyle{plain}
\bibliography{biblio.bib}
\newpage
\appendix
\section*{Appendix}

\section{Further Related Work}

Zelikovsky~\cite{zelikovsky1997series} provided the first approximation algorithm for Directed Steiner Tree Problem (\textbf{DSTP}) with factor $O(k^{\epsilon} (\log^{1/\epsilon}{k}))$ that runs in $O(n^{1/\epsilon})$, where $|U|=k$. Charikar et al.~\cite{charikar1999approximation} proposed a better approximation algorithm for \textbf{DSTP} with a factor $O(\log^3{k})$ in quasi-polynomial time. Grandoni et al.~\cite{grandoni2019log2} improved this factor and provided a randomized $O(\frac{\log^2 k}{\log \log k})$-approximation algorithm $n^{O(\log^{5} k)}$ time. Also, they showed that, unless $NP \subseteq \bigcap_{0 <\epsilon< 1} \text{ZPTIME}(2^{n^\epsilon})$ or the Projection Game Conjecture is false, there is no quasi-polynomial time algorithm for \textbf{DSTP} that achieves an approximation ratio of $o(\frac{\log^2 k}{\log \log k})$. Ghuge and Nagarajan~\cite{ghuge2020quasi} showed that their approximation algorithm results in a deterministic $O(\frac{\log^2 k}{\log \log k})$-approximation algorithm for \textbf{DSTP} in $n^{O(\log^{1+\epsilon} k)}$ time. Very recently, Li and Laekhanukit~\cite{li2022polynomial} showed that the lower bound on the
integrality gap of the flow LP is polynomial in the number of vertices.

Danilchenko et al.~\cite{danilchenko2020covering} investigated a closely related problem to \textbf{BSCP}, where the goal is to place a set of connected disks (or squares) such that the total weight of target points in the plane is maximized. They provided a polynomial-time $O(1)$-approximation algorithm for this problem. \textbf{MCSB} is also closely related to the budgeted connected dominating set problem, where the goal is to select at most $B$ connected vertices in a given undirected graph to maximize the profit function on the set of selected vertices. Khuller et al.~\cite{khuller2020analyzing} investigated this problem in which the profit function is a \emph{special submodular function}. Khuller et al.~\cite{khuller2020analyzing} designed a $\frac{12}{1-1/e}$-approximation algorithm. By generalizing the analysis of Khuller et al.~\cite{khuller2020analyzing}, Lamprou et al.~\cite{lamprou2020improved} showed that there is a $\frac{11}{1-e^{-7/8}}$-approximation algorithm for the budgeted connected dominating set problem. They also showed that for this problem we cannot achieve in polynomial time an approximation factor better than $\left(\frac{1}{1-1/e}\right)$, unless $P=NP$. 

Lee and Dooly~\cite{lee1996algorithms} provided a $(B-2)$-approximation algorithm for \textbf{URAT}, where each vertex costs $1$. Zhou et al.~\cite{zhou2018relay} studied a variant of \textbf{E-URAT} in the wireless sensor networks and provided a $10$-approximation algorithm. Seufert et al.~\cite{seufert2010bonsai} investigated a special case of the unrooted version of \textbf{URAT}, where each vertex has cost $1$ and we aim to find a tree with at most $B$ nodes maximizing the accumulated prize. This coincides with the unrooted version of \textbf{E-URAT} when the cost of each edge is $1$ and we are looking for a tree containing at most $B-1$ edges to maximize the accumulated prize. Seufert et al.~\cite{seufert2010bonsai} provided a $(5+\epsilon)$-approximation algorithm for this problem. Similarly, Huang et al.~\cite{huang2019maximizing} investigated this variant of \textbf{E-URAT} (or \textbf{URAT}) in the plane and proposed a $2$-approximation algorithm.

The quota variant of \textbf{URAT} also has been studied, which is called \textbf{Q-URAT}. Here we wish to find a tree including a vertex $r$ in a way that the total cost of the tree is minimized and its prize is no less than some \emph{quota}. By using Moss and Rabbani~\cite{moss2007approximation}'black box and the ideas of K{\"o}nemann et al.~\cite{konemann2013lmp}, and Bateni et al.~\cite{bateni2018improved}, we have an $O(\log{n})$-approximation algorithm for \textbf{Q-URAT}. This bound is tight~\cite{moss2007approximation}. The edge cost variant of \textbf{Q-URAT}, called \textbf{EQ-URAT}, has been investigated by Johnson et al.~\cite{johnson2000prize}. They showed that by adapting an $\alpha$-approximation algorithm for the $k$-MST problem, one can have an $\alpha$-approximation algorithm for \textbf{EQ-URAT}. Hence, the $2$-approximation algorithm of Garg~\cite{garg2005saving} for the $k$-MST problem results in a $2$-approximation algorithm for \textbf{EQ-URAT}.

The prize collecting variants of \textbf{URAT} have also been studied.  K{\"o}nemann et al.~\cite{konemann2013lmp} provided a Lagrangian multiplier preserving $O(\ln{n})$-approximation algorithm for \textbf{NW-PCST}, where the goal is to minimize the cost of the nodes in the resulting tree
plus the penalties of vertices not in the tree. Bateni et al.~\cite{bateni2018improved} considered a more general case of \textbf{NW-PCST} and provided an $O(\log{n})$-approximation algorithm. There exists no $o(\ln{n})$-approximation algorithm for \textbf{NW-PCST}, unless $NP \subseteq \text{DTIME}(n^{\text{Polylog}(n)})$~\cite{klein1993nearly}. The edge cost variant of \textbf{NW-PCST} has been investigated by Goemans and Williamson~\cite{goemans1995general}. They provided a $2$-approximation algorithm for \textbf{EW-PCST}. Later, Archer et al.~\cite{archer2011improved} proposed a $(2-\epsilon)$-approximation algorithm for \textbf{EW-PCST} which was an improvement upon the long standing bound of $2$.

\begin{table}[h]
\begin{center}
\begin{tabular}{ |c|c| } 
 \hline
 Problem & Best Bound \\ \hline
 \textbf{STO} & $O(\frac{\log{n}}{\log{\log{n}}})$~\cite{ghuge2020quasi} (tight) \\ \hline
 \textbf{DTO} & $O(\frac{\log{n}}{\log{\log{n}}})$~\cite{ghuge2020quasi} (tight) \\ \hline
 \textbf{DSTP} & $O(\frac{\log^2{k}}{\log{\log{k}}})$~\cite{ghuge2020quasi, grandoni2019log2} (tight) \\ \hline
 \textbf{NW-PCST} & $O(\log{n})$~\cite{bateni2018improved, konemann2013lmp} (tight) \\ \hline
  \textbf{EW-PCST} & $2-\epsilon$~\cite{archer2011improved} \\ \hline
 \textbf{URAT} & $(1+\epsilon, O(\frac{\log{n}}{\epsilon^2}))$~\cite{bateni2018improved, konemann2013lmp, moss2007approximation} \\ \hline
 \textbf{E-URAT} & $2$~\cite{paul2020budgeted} \\ \hline
 \textbf{Q-URAT} & $O(\log{n})$~\cite{bateni2018improved, konemann2013lmp, moss2007approximation} (tight)\\ \hline
 \textbf{EQ-URAT} & $2$~\cite{garg2005saving, johnson2000prize} \\ \hline
\end{tabular}
\caption{A summary of the best bounds on some variants of prize collecting problems.}\label{tbBestBounds}
\end{center}
\end{table}

\section{Some Variants of \PName}\label{apx:variants}

\paragraph*{Additive prize function}

We consider a the special case of \PName where the prize function $p$ is an additive function, i.e., for any $S \subseteq V$, $p(S)=\sum_{v \in S} p(\{v\})$, which we call \emph{Directed Rooted Additive Tree} (\textbf{DRAT}).

We apply a variant of Algorithm~\ref{algNW-DRBP} to \textbf{DRAT} where, in Line 12 of the algorithm, we use a different trimming process to achieve a better approximation factor.
In particular, we show that we can use the trimming process by Bateni et al.~\cite{bateni2018improved} (Lemma~\ref{lmBateniTrimmingProcess}) in directed graphs and achieve the same guarantee.
%We apply algorithm~\ref{algNW-DRBP} to \textbf{NW-DRB} but we have two changes: (i) in Line 4 of the algorithm, we instead define an instance $I^{Knap}_u$ of the Knapsack problem with elements $V_u$, budget $\lfloor \sqrt{B} \rfloor+1$ and profit $p(v)$, for each $v\in V_u$, and then apply an $(1-\epsilon)$-approximation algorithm to $I^{Knap}_u$ in Line 5 and, (ii) in Line 12 of the algorithm, we instead use a different trimming process to achieve a better approximation factor. To do this, we show that Lemma~\ref{lmBateniTrimmingProcess} can be applied to directed graphs and results in the same guarantee.

Let us first briefly explain the trimming process by Bateni et al.~\cite{bateni2018improved}. We are given an undirected graph $G=(V, E)$, a distinguished vertex $r \in V$ and a budget $B$, where each vertex $v \in V$ is assigned with a prize $p'(v)$ and a cost $c'(v)$. Let $T$ be a tree, where $V(T) \subseteq V$ and $E(T) \subseteq E$, let $c'(T)=\sum_{v \in V(T)} c'(v)$, and let $p'(T)=\sum_{v \in V(T)} p'(v)$. The trimming procedure by Bateni et al. takes  as input a subtree $T$ rooted at a node $r$ of a $B$-proper graph $G$ and first (i) computes a subtree $T''$ of $T$, not necessarily rooted $r$, such that $\frac{\epsilon B}{2}\leq c'(T'') \leq \eps B$ and $p'(T'')\geq \frac{\eps B}{2}\gamma$, where $\gamma = \frac{p'(T)}{c'(T)}$. Then (ii) it connects node $r$ to the root of $T''$ with a minimum-cost path and obtains a tree $T'$ rooted at $r$. Since $G$ is $B$-proper, this path exists and has length at most $B$, which implies that the cost of the resulting tree is between $\frac{\epsilon B}{2}$ and $(1+\eps)B$, and its prize to cost ratio is at least $\frac{p'(T'')}{(1+\eps)B}\geq\frac{\epsilon \gamma}{4}$. This leads to Lemma~\ref{lmBateniTrimmingProcess}. For the case of directed graphs, we prove a lemma equivalent to Lemma~\ref{lmBateniTrimmingProcess}.

\begin{lemma}\label{coTrimmingProcess}
Let $D= (V, A)$ be a $B$-appropriate graph for a node $r$. Let $T$ be an out-tree of $D$ rooted at $r$ with the prize-to-cost ratio $\gamma=\frac{p(T)}{c(T)}$. Suppose that for $\epsilon \in (0, 1]$, $c(T) \ge \frac{\epsilon B}{2}$. One can find an out-tree $\hat{T}$ rooted at $r$ with the prize-to-cost ratio at least $\frac{\epsilon \gamma}{4}$ such that $\epsilon B/2 \le c'(T') \le (1+\epsilon)B$.
\end{lemma}
\begin{proof}
\iffalse
The trimming process used by Bateni et al. computes a subtree $T''$ of the original tree $T$ with $c'(T'') \leq \eps B$ and $p'(T'')\geq \frac{\eps B}{2}\gamma$ and then connects node $r$ to the root of $T''$ in order obtain $T'$. Since the underlying graph is $B$-proper, this path exists and has length at most $B$, which implies that the cost of the resulting tree is at most $(1+\eps)B$ and the prize to cost ratio is at least $\frac{p(T'')}{(1+\eps)B}\geq\frac{\epsilon \gamma}{4}$. 
\fi
We first define an undirected tree by removing the directions from $T$ and apply to it the step (i) of the trimming procedure by Bateni et al. with $p'(\cdot) = p(\cdot)$ and $c'(\cdot) =c(\cdot)$ as prize and cost functions, respectively. We obtain a subtree $T''$ such that $\frac{\epsilon B}{2}\leq c'(T'') \leq \eps B$ and $p'(T'')\geq \frac{\eps B}{2}\gamma$. We reintroduce in $T''$ the same directions as in $T$  to obtain an out-tree $T''_{out}$ corresponding to $T''$. Then, we add a minimum cost path from $r$ to the root of $T''_{out}$. Since $D$ is $B$-appropriate, this path exists and has length at most $B$. The obtained out-tree $\hat{T}$ has the desired properties because $p(T''_{out}) = p'(T'') $ and $ c(T''_{out}) = c'(T'')$ and therefore $\frac{\epsilon B}{2}\leq c(\hat{T})\leq (1+\eps)B$ and $\frac{p(\hat{T})}{c(\hat{T})}\geq\frac{p'(T'')}{(1+\eps)B}\geq\frac{\epsilon \gamma}{4}$.
\end{proof}

This results in the following theorem.

\begin{theorem}\label{thNWDRBAlgorithm1}
There exists a polynomial-time bicriteria $(1+\epsilon, O(\frac{\sqrt{B}}{\epsilon^2}))$-approximation algorithm for \textbf{DRAT}.
\end{theorem}
\begin{proof}[Sketch of proof]
The theorem follows by the same arguments used in the proof of Theorem~\ref{thMain} and by Lemma~\ref{coTrimmingProcess}.
\end{proof}

\paragraph*{Undirected graphs}

%Maximum Connected Submodular function with Budget constraint (\textbf{MCSB}).

Here we consider the particular case of \PName in which the graph is undirected. Since it is a special case of \PName, all our results hold.

Kuo et al.~\cite{kuo2014maximizing} studied the unrooted version of \PName in undirected graphs, which they call \emph{Maximum Connected Submodular function with Budget constraint} (\textbf{MCSB}).
In \textbf{MCSB}, we are given an undirected graph $G=(V, E)$, a monotone submodular function $p: 2^V \rightarrow \mathbb{R}^+$, a cost function $c: V \rightarrow \mathbb{Z}^+$ and a budget $B$, and the goal is to find a tree $T$ s.t. $c(T)=\sum_{v \in V(T)}c(v) \le B$ and $p(T)$ is maximum.

Kuo et al. provided an $O((\Delta+1) \sqrt{B})$-approximation algorithm for \textbf{MCSB}, where $\Delta$ is the maximum degree. Recently, Hochbaum and Rao~\cite{hochbaum2020approximation} investigated \textbf{MCSB} in the case in which each vertex costs $1$ for the detection of mutated driver pathways in cancer and provided an approximation algorithm with factor $\min\{\frac{1}{(1-1/e)(1/R-1/B)}, B\}$, where $R$ is the radius of the input graph $G$.

Our algorithm for \textbf{DUST} can be applied to undirected graphs and improves the above results~\cite{kuo2014maximizing,hochbaum2020approximation} by providing an $O(\sqrt{B})$-approximation.

\paragraph*{Additive prize function in undirected graphs}
We consider the problem where we are given an undirected graph $G=(V,E)$, a cost function $c:V \rightarrow \mathbb{R}^+$, a prize function $p:V \rightarrow \mathbb{R}^+$, a budget $B$ and a vertex $r$, and the goal is to find a tree $T$ such that $c(T)=\sum_{v \in V(T)}c(v) \le B$, $r \in V(T)$ and $p(T)=\sum_{v \in V(T)}p(v)$ is maximum. 
We call this problem \emph{Undirected Rooted Additive Tree} (\textbf{URAT}) and we show that it admits a bicriteria polynomial-time approximation algorithm, which depends on the maximum degree in the graph.

\begin{theorem}\label{thRBPDegree}
\textbf{URAT} admits a polynomial-time bicriteria $(1+\epsilon, \frac{16 \Delta}{\epsilon^2})$-approximation algorithm, where $\Delta$ is the maximum degree in the given graph.
\end{theorem}

To prove Theorem~\ref{thRBPDegree}, we need the following observations. Let \textbf{E-URAT} be the same problem as \textbf{URAT} except that the edges and not the nodes have costs. In particular, in \textbf{E-URAT}, we are given an undirected graph $G=(V,E)$, a cost function $c:E \rightarrow \mathbb{R}^+$, a prize function $p:V \rightarrow \mathbb{R}^+$, a budget $B$ and a vertex $r$, and the goal is to find a tree $T=(V(T), E(T))$ such that $c(T)=\sum_{e \in E(T)}c(e) \le B$, $r \in V(T)$ and $p(T)=\sum_{v \in V(T)}p(v)$ is maximum.
Paul et al.~\cite{paul2020budgeted} gave a $2$-approximation algorithm for \textbf{E-URAT}.

\begin{lemma}\label{lmEW-RBPandRBP}
Let $\mathcal{A}$ be an $\alpha$-approximation algorithm for \textbf{E-URAT}. Using $\mathcal{A}$, a bicriteria $(\Delta, \alpha)$-approximation can be obtained for \textbf{URAT}, where $\Delta$ is the maximum degree of a node.
\end{lemma}

\begin{proof}
Consider an instance $I=\langle G= (V, E), p, c, r, B \rangle$ of \textbf{URAT}. We create an instance $I'=\langle G, p, c', r, B' \rangle$ of \textbf{E-URAT} where $B'=\Delta B$  and, for any edge $e=\{u,v\} \in E'$, we set $c'(e)=c(u)+c(v)$. Let $OPT_{I}$ and $OPT_{I'}$ be optimal solutions to $I$ and $I'$, respectively.

%\begin{observation}\label{obsNW-RBPandEW-RBP}
%$p(OPT_{I'}) \ge p(OPT_{I})$.
%\end{observation}
%\begin{proof}
%From any solution $T$ to $I$ where $c(T)\le B$, we can always create a solution $T'$ to $I'$ as above, where $c(T')\le \Delta B$ and $p(T')=p(T)$. This means that we never have the case where $p(OPT_{I'})< p(OPT_{I})$, otherwise by transforming $OPT_{I}$ to a solution for $I'$, we can achieve a better solution for $I'$, which contradicts the assumption that $OPT_{I'}$ is an optimal solution to $I'$. This concludes the proof of the observation.
%\end{proof}

First, observe that $p(OPT_{I'}) \ge p(OPT_{I})$ as $OPT_{I}$ is a feasible solution for $I'$. By running $\mathcal{A}$ on instance $I'$, we have a tree $T'$ with the cost $c(T')=\sum_{e \in E(T')} c(e) \le B'=\Delta B$ and the prize $p'(T') \ge \alpha p(OPT_{I'})\ge \alpha p(OPT_{I})$. Therefore $T'$ is a $(\Delta, \alpha)$-approximation  for $I$.
%From $T'$, we can obtain a solution $T=(V(T), E(T))$ to \textbf{URAT} as follows. We set $V(T)=V(T')$ and $E(T)=E(T')$. For each $v \in V(T)$, we set $p(v)=p'(v)$ and $c(v)=c'(\{u,v\})-c(u)$. This concludes the proof.
\end{proof}

%Now consider an instance $I=\langle G= (V, E), p, c, r, B \rangle$ of \textbf{URAT}. A graph $G$ is called $B$-proper for the vertex $r$ if the cost of reaching any vertex from $r$ is at most $B$. Consider a subtree $T=(V(T), E(T))$ of $G$, where $V(T) \subseteq V(G)$ and $E(T) \subseteq E(G)$. Let $c(T)=\sum_{v \in V(T)} c(v)$ and $p(T)=\sum_{v \in V(T)} p(v)$. Let $\gamma=\frac{p'(T)}{c'(T)}$ be the prize-to-cost ratio of $T$. Bateni, Hajiaghayi and Liaghat~\cite{bateni2018improved} proposed a trimming process that leads to the following.

%\begin{lemma}[Lemma 3 in \cite{bateni2018improved}]\label{lmBateniTrimmingProcess}
%Let $T$ be a subtree rooted at $r$ with the prize-to-cost ratio $\gamma$. Suppose the underlying graph is $B$-proper for $r$ and for $\epsilon \in (0, 1]$ the cost of the tree is at least $\frac{\epsilon B}{2}$. One can find a tree $T^*$ containing $r$ with the prize-to-cost ratio at least $\frac{\epsilon \gamma}{4}$ such that $\epsilon B/2 \le c'(T^*) \le (1+\epsilon)B$.
%\end{lemma}

Now we are ready to prove Theorem~\ref{thRBPDegree}.

\begin{proof}[Proof of Theorem~\ref{thRBPDegree}]
Consider an instance $I=\langle G= (V, E), p, c, r, B \rangle$ of \textbf{URAT}. From graph $G$, we first create a maximal inclusion-wise $B$-appropriate graph $G'$ for $r$. Let $OPT$ be an optimal solution to $I$. By Lemma~\ref{lmEW-RBPandRBP} and the $2$-approximation algorithm for \textbf{E-URAT} proposed by Paul et al.~\cite{paul2020budgeted}, we create a subtree $T'$ of $G'$ with prize $p(T)\ge \frac{OPT}{2}$ and cost $c(T) \le \Delta B$. Now by using Lemma~\ref{lmBateniTrimmingProcess}, we obtain a subtree $\hat{T}$ from $T$ with the cost at most $c(\hat{T}) \le (1+\epsilon) B$ and the prize-to-cost ratio:
\[
\frac{p(\hat{T})}{c(\hat{T})} \ge \frac{\epsilon p(T)}{4c(T)} \ge \frac{\epsilon}{8 \Delta B}OPT_{B}.
\]
As $c(\hat{T}) \ge \epsilon B/2$, we have $p(\hat{T}) \ge \frac{\epsilon^2}{16 \Delta} OPT_{B}$, which concludes the proof.
\end{proof}

\paragraph*{Quota Problem}

%This problem was formally stated in the Related Work section, which is the undirected version of \textbf{NW-DRQP}. Then, Theorem~\ref{thDirctedRootedQuota} can be applied to the rooted quota problem.

%\begin{theorem}\label{thNW-RQPGhuge}
%There is a bicritria $O(\frac{\log \log k}{\log k}, 1+\beta)$-approximation algorithm for \textbf{NW-RDQP} that runs in $(n \log \frac{C}{c_{min}} \log B)^{O(\log^{1+\epsilon} k)}$ time for any constant $\epsilon, \beta > 0$, where $k \le n=|V|$ is the number of vertices in an optimal solution.
%\end{theorem}
We consider the quota version of \textbf{URAT}, called \textbf{Q-URAT},  where we are given an undirected graph $G=(V,E)$, a cost function $c:V \rightarrow \mathbb{R}^+$, a prize function $p:V \rightarrow \mathbb{R}^+$, a quota $Q\in \mathbb{R}^+$, and a vertex $r$, and the goal is to find a tree $T$ such that $p(T)=\sum_{v \in V(T)} p(v)\ge Q$, $r \in V(T)$ and $c(T)=\sum_{v \in V(T)}c(v)$ is minimum.

We provide a new polynomial-time approximation algorithm for \textbf{Q-URAT} that depends on the maximum degree of the given graph.

\begin{theorem}\label{thNW-RQPDegree}
\textbf{Q-URAT} admits a $2\Delta$-approximation algorithm, where $\Delta$ is the maximum degree of a node in the given graph.
\end{theorem}

To prove Theorem~\ref{thNW-RQPDegree}, we first prove the following lemma. Let \textbf{EQ-URAT} be the same problem as \textbf{Q-URAT} except that the edges and not the nodes have costs. In particular, in \textbf{EQ-URAT}, we are given an undirected graph $G=(V,E)$, a cost function $c:E \rightarrow \mathbb{R}^+$, a prize function $p:V \rightarrow \mathbb{R}^+$, a quota $Q$ and a vertex $r$, and the goal is to find a tree $T$ such that $p(T)=\sum_{v \in V(T)} p(v)\ge Q$, $r \in V(T)$ and $c(T)=\sum_{e \in E(T)}c(v)$ is minimum.

\begin{lemma}\label{lmEW-RQPandNW-RQP}
Let $\mathcal{A}$ be an $\alpha$-approximation algorithm for \textbf{EQ-URAT}. Using $\mathcal{A}$, an $\alpha \Delta$-approximation can be obtained for \textbf{Q-URAT}.
\end{lemma}

\begin{proof}
From an instance $I=\langle G= (V, E), p, c, r, Q \rangle$ of \textbf{Q-URAT}, we create an instance $I'=\langle G, p, c', r, Q \rangle$ of \textbf{EQ-URAT}, where, for any edge $e=\{u,v\} \in E'$, we set $c'(e)=c(u)+c(v)$. 

By running $\mathcal{A}$ on instance $I'$, we have a tree $T$ with  prize $p(T) \ge Q$ and cost $c'(T)\le \alpha c'(OPT_{I'})$, where $OPT_{I'}$ is an optimal solution to $I'$.
Moreover, we have that $c(T)\le c'(T)$ and $c'(OPT_{I}) \le \Delta c(OPT_{I})$, where $OPT_{I}$ is an optimal solution to $I$.
We obtain $c(T) \le \alpha \Delta c(OPT_{I})$. Therefore $T$ is an $\alpha \Delta$-approximation for $I$.
\end{proof}

Now we prove Theorem~\ref{thNW-RQPDegree}.

\begin{proof}[Proof of Theorem~\ref{thNW-RQPDegree}]
Johnson et al.~\cite{johnson2000prize} observed that an $\alpha$-approximation algorithm for the $k$-MST problem guarantees an $\alpha$-approximation also for \textbf{EQ-URAT}. This observation, along with the $2$-approximation algorithm by Garg~\cite{garg2005saving} for the $k$-MST problem result in a $2$-approximation algorithm for \textbf{EQ-URAT}. Therefore, by Lemma~\ref{lmEW-RQPandNW-RQP}, the theorem follows.
\end{proof}

\paragraph*{Maximum Weighted Budgeted Connected Set Cover (\textbf{MWBCSC}).}

Let $X$ be a set of elements, $\mathcal{S} \subseteq 2^X$ be a collection sets, $p: X \rightarrow \mathbb{R}^+$ be a prize function, $c: \mathcal{S} \rightarrow \mathbb{R}^+$ be a cost function, $G_{\mathcal{S}}$ be a graph on vertex set $\mathcal{S}$, and $B$ be a budget. In \textbf{MWBCSC}, the goal is to find a subcollection $\mathcal{S}'\subseteq \mathcal{S}$ such that $c(\mathcal{S}')=\sum_{S \in \mathcal{S}'}c(S) \le B$, the subgraph induced by $\mathcal{S}'$ is connected and $p(\mathcal{S}')=\sum_{x \in X_{\mathcal{S}'}} p(x)$ is maximized, where $X_{\mathcal{S}'}=\bigcup_{S \in \mathcal{S}'} S$. Note that, as \textbf{MWBCSC} coincides with the unrooted version of \textbf{URAT} when no two sets share an element, \textbf{MWBCSC} admits no $o(\log \log n)$-approximation algorithm unless $NP \subseteq DTIME(n ^{\text{polylog}(n)})$ even if the budget constraint is violated by a factor of any universal constant $\rho$~\cite{kortsarz2009approximating}. Moreover, since the prize function of \textbf{MWBCSC} is submodular w.r.t. subsets of nodes, our algorithm for the undirected unrooted version of $\PName$ guarantees an $O(\sqrt{B})$-approximation for \textbf{MWBCSC}.

Consider the following assumption.
\begin{assumption}\label{assumAdjacencyMWBCSC}
Two sets $S$ and $S'$ having an element in common are adjacent in $G_{\mathcal{S}}$.
\end{assumption}

Denote by \textbf{UUAT} the unrooted version of \textbf{URAT}. Ran et al.~\cite{ran2016approximation} provided a $\frac{2(\Delta+1)\alpha}{1-e^{-1}}$-approximation algorithm for \textbf{MWBCSC} under Assumption~\ref{assumAdjacencyMWBCSC}, where $\alpha$ is the performance ratio for the algorithm used to approximate \textbf{UUAT}. Note that the best known algorithm for \textbf{UUAT} is an $O(\log{n})$-approximation~\cite{bateni2018improved}. Let $f$ be the maximum frequency of an element in an instance of \textbf{MWBCSC}. Clearly, in the work by Ran et al.~\cite{ran2016approximation}, $\Delta \ge f$ as the authors investigated \textbf{MWBCSC} under Assumption~\ref{assumAdjacencyMWBCSC}. In the following, we aim to provide some approximation results for some variants of \textbf{MWBCSC}.

We first show that there exists a polynomial-time $\alpha f$-approximation algorithm for \textbf{MWBCSC} (without Assumption~\ref{assumAdjacencyMWBCSC}), where $\alpha$ is the performance ratio of an algorithm for \textbf{UUAT}.

\begin{theorem}\label{thMWBCSCGeneral}
\textbf{MWBCSC} admits a polynomial-time $\alpha f$-approximation algorithm, where $f$ is the maximum frequency of an element and $\alpha$ is the performance ratio of an algorithm for \textbf{UUAT}.
\end{theorem}
\begin{proof}
 Let $I=(X, \mathcal{S}, c, p, G_{\mathcal{S}}=(\mathcal{S},E_S), B)$ be an instance \textbf{MWBCSC}. We create an instance $J=(G_{\mathcal{S}}, c, p', B)$ of \textbf{UUAT}, where $p'(S)=\sum_{x \in S} p(x)$ for each $S\in \mathcal{S}$, that results in loosing an approximation factor $f$. 
 
 Let $T$ be the solution returned by an $\alpha$-approximation algorithm for \textbf{UUAT} applied to $J$, $X_{T}=\cup_{S\in V(T)}S$, $OPT_I$, and $OPT_j$ be optimal solutions for $I$ and $J$, respectively.
 
 We have $p'(T)\ge \frac{1}{\alpha}p'(OPT_J)$, $p(OPT_I)\le p'(OPT_I)\le p'(OPT_J)$, and
\begin{align*}
    p'(T)= \sum_{S\in V(T)} \sum_{x\in S} p(x) =\sum_{x\in X_{T}}|\{S\in V( T) : x\in S\}|p(x) \leq f\sum_{x\in X_{T}}p(x) = f p(V(T)).
\end{align*}

Therefore,
\begin{align*}
    p(V(T))\ge \frac{1}{f} p'(T) \ge \frac{1}{\alpha f}p'(OPT_J)\ge \frac{1}{\alpha f}p'(OPT_I).
\end{align*}
This implies that $V(T)$ is an $\alpha f$-approximation for $I$.
\end{proof}
Since the algorithm by Bateni et al. for \textbf{UUAT} guarantees an $O(\log n)$-approximation, where $n$ is the number of nodes in the graph, then \textbf{MWBCSC} admits a polynomial-time $O(f\log{|\mathcal{S}|})$-approximation algorithm.

In the next step, we show that \textbf{MWBCSC} admits a $2f$-approximation algorithm when each set has cost $1$. Note that this variant of \textbf{MWBCSC} has some applications, e.g., the detection of mutated pathways in cancers~\cite{hochbaum2020approximation, vandin2011algorithms} and wireless sensor networks~\cite{gao2018algorithm, huang2020approximation, kuo2014maximizing, ran2016approximation, xu2021throughput, yu2019connectivity}.

\begin{theorem}\label{thMWBCSCwithCost1}
When each set costs $1$, \textbf{MWBCSC} admits a polynomial-time $2f$-approximation algorithm, where $f$ is the maximum frequency of an element.
\end{theorem}
\begin{proof}
We observe that \textbf{URAT} when each vertex costs $1$ is equivalent to \textbf{E-URAT} when each edge costs $1$ from an approximation point of view. Indeed, each tree with at most $B-1$ edges includes at most $B$ nodes and vice-versa. The same holds for the unrooted versions of \textbf{URAT} and \textbf{E-URAT}, which we call \textbf{UUAT} and \textbf{E-UUAT}, respectively. Moreover, \textbf{E-UUAT} admits a $2$-approximation algorithm that consists in guessing the root and applying the 2-approximation algorithm for \textbf{E-URAT} by Paul et al.~\cite{paul2020budgeted}. Therefore, we have a $2$-approximation algorithm for \textbf{UUAT} when each vertex costs 1 and the theorem follows by Theorem~\ref{thMWBCSCGeneral}.
\end{proof}

In the next step, we will provide an approximation algorithm for \textbf{MWBCSC} which does not depend on $f$. Recall that in a node-weighted graph $G$, for any two vertices $u$ and $v$, $dist(u,v)$ is the minimum cost of a path $P$ from $u$ to $v$, where the cost of $P$ is the sum of the costs of its nodes. %Now Consider the following assumption.
%\begin{assumption}\label{assumDistanceMWBCSC}
%Two vertex sets $S, S' \in \mathcal{S}$ in an instance $I=(X, \mathcal{S}, c, p, G_{\mathcal{S}}, B)$ of \textbf{MWBCSC} have an element in common, i.e., $S \cap S' \ne \emptyset$, if $dist(S,S')-(c(S)+c(S)) \le h$, where $h\ge 1$ and is a real number.
%\end{assumption}
Now we show the following.

\begin{theorem}\label{thMWBCSCDistance}
Let $\mathcal{A}$ be an $\alpha$-approximation algorithm for \textbf{UUAT}. \textbf{MWBCSC} admits a polynomial-time bicriteria $(h+1, \alpha)$-approximation algorithm, where $h=\max_{\substack{S, S' \in \mathcal{S}\\S \cap S' \ne \emptyset}} (dist(S,S')-(c(S)+c(S')))$ such that $c(S) \ge 1, \forall{S \in \mathcal{S}}$.
\end{theorem}

To prove Theorem~\ref{thMWBCSCDistance}, we need some observations and materials.
Let $I=(X, \mathcal{S}, c, p, G_{\mathcal{S}}, B)$ be an instance \textbf{MWBCSC}. Now we show that one can transform an instance $I$ of \textbf{MWBCSC} to an instance of $J=(G'=(V', E'), c', p', B)$ of \textbf{UUAT}. From $G_{\mathcal{S}}=(V(G_{\mathcal{S}}), E(G_{\mathcal{S}}))$, we create graph a $G'=(V', E')$ as follows: $V'=V_{\mathcal{S}} \cup V_X$, where $V_{\mathcal{S}}=\{v_S:S \in V(G_{\mathcal{S}})\}$ and  $V_X=\{v_x:x\in X\}$, and $E'= E(G_{\mathcal{S}}) \cup \{\{v_S, v_x\}:x\in S\}$. We set for any $v_S \in V_{\mathcal{S}}$, $p'(v_S)=0$ and $c'(v_S)=c(S)$. For any $v_x \in V_X$, $p'(v_x)= p(x)$ and $c'(v_x)=0$. % Last, For any $v_e \in V' \cap V_{E_1}$, $p'(v_e)= 0$ and $c'(v_e)=0$.

\begin{lemma}\label{lmSolutionFromNWBPtoMWBCSCwithDistance}
Let $SOL_J$ be a solution to \textbf{UUAT} on the instance $J$. From $SOL_J$, a solution $SOL_I$ to \textbf{MWBCSC} on instance $I$ can be obtained such that $p(SOL_I) \ge p'(SOL_J)$ and $c(SOL_I) \le (h+1)B$.
\end{lemma}
\begin{proof}
Let $Sol'=\emptyset$. We repeat the following process: if there exist two edges $\{v_S, v_x\}, \{v_{S'}, v_x\} \in E'(SOL_J)$, then $Sol'=Sol'\cup \{v_S, v_{S'}\}$, where $\{v_S, v_{S'}\}$ is an edge. Let $G(Sol')$ be the graph with edges in $Sol'$. Let $\bar T=(V(\bar T), E(\bar T))$ be a spanning tree of $G(Sol')$. Note that we have $|E(\bar T)| \le B-1$ as $c(S) \ge 1, \forall{S \in \mathcal{S}}$. 
Note that if there exist two edges $\{v_S, v_x\}, \{v_{S'}, v_x\} \in E'(SOL_J)$, where $v_S \ne v_{S'}$, there is a path $P$ of distance at most $h+c(v_S)+c(v_{S'})$ between $v_S, v_{S'} \in G_{\mathcal{S}}$. Let $SP_e$ be a shortest path on $G_{\mathcal{S}}$ connecting $e=\{v_S, v_{S'}\}$. Let $Sol''=\{SP_e: e \in E(\bar T)\}$ and $G(Sol'')$ is a graph with edges in $Sol''$. $G(Sol'')$ is a connected subgraph of $G_{\mathcal{S}}$. Let $SOL_I$ be a spanning tree of $G(Sol'')$. This means that $c(SOL_I) \le B+(B-1)h \le (h+1)B$ as $c(\bar T) \le B$ and $|E(\bar T)| \le B-1$ and each edge $e$ of $G(Sol')$ is replaced by a path $SP_e$ of cost at most $h$ in $SOL_I$. Also, note that $p(SOL_J)=\sum_{v_x \in V_X \cap V'(SOL_J)} p(v_x)$ as vertices in $ V_{\mathcal{S}} \cap V'$ has prize zero in $J$. So, we have $p'(SOL_I) \ge p(SOL_J)$ as $SOL_I$ spans all vertices in $Sol''$ and $Sol''$ spans all vertices in $Sol'$. Hence a feasible solution (with at least the same prize as $SOL_J$ and cost no more than $(h+1)B$) can be obtained from $SOL_J$ for \textbf{MWBCSC}. This concludes the lemma.
\end{proof} %By claim~\ref{clSolutionFromNWBPtoMWBCSCunderAssumtion},

\begin{observation}\label{obsDistanceNW-BP-MWBCSCD}
Let $J=(G'=(V', E'), c', p', B)$ be an instance of \textbf{UUAT} resulting in our transformation from an instance $I=(X, \mathcal{S}, c, p, G_{\mathcal{S}}, B)$ of \textbf{MWBCSC}. Let $OPT_J$ and $OPT_I$ be optimal solutions to $J$ and $I$, respectively. Then $p(OPT_J) \ge p(OPT_I)$.
\end{observation}

\begin{proof}
From any solution $T$ to $I$ where $c(T)\le B$, we can always create a solution $T'$ to $J$ by our transformation, where $c(T')=c(T)=B$ and $p(T')=p(T)$. This means that we never have the case where $p(OPT_J)< p(OPT_I)$, otherwise by transforming $OPT_I$ to a solution for $J$, we can achieve a better solution for $J$, which contradicts the assumption that $OPT_J$ is an optimal solution to $J$. This concludes the proof.\qedhere

%Let $T=(V(T), E(T))$ be an optimal solution to \textbf{UUAT} on the instance $J$. Note that if there exist two edges $\{v_S, v_x\}, \{v_{S'}, v_x\} \in E(T)$, it costs at most $h$ to connect $v_S$ and $v_{S'}$ in a solution to \textbf{MWBCSC}. Since $B'=B$ (by our transformation), meaning that from $T=(V(T), E(T))$, we may not be able to obtain a feasible solution to \textbf{MWBCSC} which respects the budget constraint. This implies that we cannot cover some vertex set, so do their elements, by a budget $B'$. This finishes the proof.
\end{proof}

Now we are ready to prove Theorem~\ref{thMWBCSCDistance}.

\begin{proof}[Proof of Theorem~\ref{thMWBCSCDistance}]
Let $OPT_J$ and $OPT_I$ be the optimum solutions to $J$ and $I$, respectively. By Lemma~\ref{lmSolutionFromNWBPtoMWBCSCwithDistance} and Observation~\ref{obsDistanceNW-BP-MWBCSCD}, any $\alpha$-approximation for \textbf{UUAT} results in a solution $SOL_I$ for \textbf{MWBCSC} such that $p(SOL_I) \ge \frac{OPT_J}{\alpha} \ge \frac{OPT_I}{\alpha}$ and $c(SOL_I) \le (h+1)B$. %By Lemma~\ref{lmBateniTrimmingProcess}, from $SOL_I$ we can achieve another solution $T'$ such that $p(T')\ge \frac{OPT_I}{(h+1)\alpha}$ and $c(T')\le B$. This concludes the proof.
\end{proof}

\begin{corollary}\label{cor:logMWBCSCD}
One can have a polynomial-time $O(h\log{n})$-approximation algorithm for \textbf{MWBCSC}, when $c(S) \ge 1, \forall{S \in \mathcal{S}}$.
\end{corollary}

\begin{proof}
%Let $J=(G'=(V', E'), c', p', B)$ be an instance of \textbf{UUAT} resulting in our transformation from an instance $I=(X, \mathcal{S}, c, p, G_{\mathcal{S}}, B)$ of \textbf{MWBCSC}. 

Note that Bateni et al.~\cite{bateni2018improved}'s approach for the prize collecting problem and the black-box introduced by Moss and Rabbani~\cite{moss2007approximation} result in an $O(\log n)$-approximation algorithm for \textbf{UUAT} violating the budget constraint by a factor of $2$. We call this approach \textbf{BMAlg}. Let $T$ be such a solution after applying \textbf{BMAlg}, where $c(T)\le 2B$ and $p(T)=O(p(OPT))$ that $OPT$ is an optimal solution to \textbf{UUAT}. Then Bateni et al.~\cite{bateni2018improved} provided a trimming process to achieve a solution $T' \subseteq T$ which has cost $c(T')\le B$ by loosing some constant in the approximation factor.

Note that as in \textbf{MWBCSC} the prize function is monotone submodular, we need to use our trimming process (i.e., Lemma~\ref{lmBateniTrimmingProcessGeneralUnrooted}). Now we follow the same argument as that of Theorem~\ref{thMainUndirectedEW-DRBP}. This means that we first restrict ourselves to flat trees by guessing the vertex $x$ with the maximum cost, setting the budget $B'=B-c(x)$ and removing vertices with cost more than $B/2$. Next we transform the instance $I$ of \textbf{MWBCSC} to the instance $J$ of \textbf{UUAT} as above. Now by \textbf{BMAlg}, we can achieve a solution $SOL_J$ to $J$ such that $p(SOL_J)=O(\log n)$ and $c(SOL_J)\le 2B$. We then transform $SOL_J$ to a solution for $I$, denoted by $SOL_I$ as Lemma~\ref{lmSolutionFromNWBPtoMWBCSCwithDistance}. Hence $p(SOL_I)=O(\log n)$ and $c(SOL_I)\le 2B(h+1)$. Finally, we apply Lemma~\ref{lmBateniTrimmingProcessGeneralUnrooted} to achieve a subtree $\hat{T}$ of $SOL_I$ with the cost $B/4 \le c(\hat{T}) \le B$ and the prize-to-cost ratio $\frac{p(\hat T)}{c(\hat T)}\ge\frac{\gamma}{32(2B(h+1))+8}$. This implies that $p(\hat{T})=O(\frac{1}{h\log n})$.\qedhere

%\[
%p(\hat{T})\ge\frac{\gamma}{64B(h+1)+8} \cdot \frac{B}{4}=O(\frac{1}{(h+1)\log n})\].\qedhere

%an out-tree $SOL_I$ to  Now by the same argument as that of Theorem~\ref{thMainUndirectedEW-DRBP}, we can have an $O((h+1)\log n)$-approximation algorithm for \textbf{MWBCSC}.
\end{proof}

Another important consequence of Theorem~\ref{thMWBCSCDistance} is that an $\alpha$-approximation algorithm for \textbf{UUAT} yields an $\alpha$-approximation algorithm for \textbf{MWBCSC} under Assumption~\ref{assumAdjacencyMWBCSC}. This is an improvement over $\frac{2(\Delta+1)\alpha}{1-e^{-1}}$-approximation algorithm of Ran et al.~\cite{ran2016approximation}, where $\Delta$ is the maximum degree of the graph. Note that in this case we do not need to assume that $c(S) \ge 1, \forall{S \in \mathcal{S}}$.

\begin{corollary}\label{assumAdjacencyCMCP}
One can have a polynomial-time $O(\log{n})$-approximation algorithm for \textbf{MWBCSC} under Assumption~\ref{assumAdjacencyMWBCSC}.
\end{corollary}

\begin{proof}
Let $J=(G'=(V', E'), c', p', B)$ be an instance of \textbf{UUAT} resulting in our transformation from an instance $I=(X, \mathcal{S}, c, p, G_{\mathcal{S}}, B)$ of \textbf{MWBCSC}. Now we run the algorithm by Bateni et al. on $J$. Let $SOL_J$ be the solution that returned by their approach, where $c(SOL_J)\le B$ and $p(SOL_J)=O(\log n)$. We know that if there exist two edges $\{v_S, v_x\}, \{v_{S'}, v_x\} \in E(SOL_J)$, then $S$ and $S'$ are neighbor in $G_{\mathcal{S}}$ (by Assumption~\ref{assumAdjacencyMWBCSC}). So, we can obtain a solution $SOL_I=(V(SOL_I), E(SOL_I))$ for \textbf{MWBCSC} form $SOL_J$ with the same cost and prize.
\end{proof}

%\begin{theorem}\label{thMWBCSCunderAssumption}
%Let $\mathcal{A}$ be an $\alpha$-approximation algorithm for \textbf{UUAT}. Using $\mathcal{A}$, an $\alpha$-approximation algorithm for \textbf{MWBCSC} under Assumption~\ref{assumAdjacencyMWBCSC} can be obtained.
%\end{theorem}

%\begin{corollary}
%By Theorems~\ref{thNW-BPofBateni} and~\ref{thMWBCSCunderAssumption}, one can have a polynomial-time $O(\log{n})$-approximation algorithm for \textbf{MWBCSC}.
%\end{corollary}

%\begin{remark}
 %In case each set has cost no more than $B/2$, by using the trimming process by Bateni, Hajiaghay and Liaghat~\cite{bateni2018improved}, our approach results in a $32\alpha$-approximation algorithm for \textbf{MWBCSC}, which is an improvement over the approximation factor $\frac{64\alpha}{1-e^{-1}}$~\cite{ran2016approximation}. 
%\end{remark}
%Using an $\alpha$-approximation algorithm for \textbf{UUAT}, Ran et al.~\cite{ran2016approximation} showed that in case each set $S\in \mathcal{S}$ has cost $c(S) \le B/2$, one can have $\frac{64\alpha}{1-e^{-1}}$-approximation algorithm for \textbf{MWBCSC} under Assumtion~\ref{assumAdjacencyMWBCSC}. In the next step, we improve this result and provide a $32\alpha$-approximation algorithm for this case.
%\begin{theorem}\label{thMWBCSCAssum1}
%Let $\alpha$ be the performance ratio for the algorithm used to approximate \textbf{UUAT}. Under Assumption~\ref{assumAdjacencyMWBCSC}, \textbf{MWBCSC} admits a polynomial-time $32\alpha$-approximation algorithm when each set $S$ has cost $c(S)\le B/2$. 
%\end{theorem}

Last we show that there is a polynomial-time constant approximation algorithm for \textbf{MWBCSC} when each set costs one under Assumption~\ref{assumAdjacencyMWBCSC}.

\begin{theorem}\label{thMWBCSC}
Under Assumption~\ref{assumAdjacencyMWBCSC}, \textbf{MWBCSC} admits a polynomial-time $\frac{16}{1-e^{-1}}$-approximation algorithm, when each set has cost one.
\end{theorem}

We need one useful observation.

\begin{lemma}[Lemma 3 in Lamprou et al.~\cite{lamprou2020improved}]\label{lmLamprou}
%Let $T=(S, F)$ be a tree with $t=|F|$ edges, where $k\le t \le 2k$. Then one can find in polynomial time at most $3$ subtrees of $T$ with at most $k$ edges each, such that their edge sets partition $F$.
Let $k$ be an integer. Given any tree $T$ on $ak$ vertices, where $a \in N$ is a constant, and $k \ge 4a-2$, we can decompose it into $2a$ subtrees each on at most $k$ vertices.
\end{lemma}

We are ready to prove Theorem~\ref{thMWBCSC}.

\begin{proof}
Let $I=(X, \mathcal{S}, c, p, G_{\mathcal{S}}, B)$ be an instance \textbf{MWBCSC}. We first assign prizes to the nodes in $G_{\mathcal{S}}$ as follows: we iteratively pick the node $v$ in $G_{\mathcal{S}}$ that covers the elements with the maximum prize $p_{max}$, assign $p_{max}$ to $v$ as its prize and remove $v$ and the covered elements from further consideration. Ran et al.~\cite{ran2016approximation} showed that this assigning-prize process to the nodes ensures that there exists a connected subcollection of $\mathcal{S}' \subseteq \mathcal{S}$ such that $p(\mathcal{S}') \ge \frac{1-e^{-1}}{2}opt$ and $c(\mathcal{S}') \le 2B$, where $opt$ is the optimum solution to \textbf{MWBCSC}. As in our case each vertex in $G_{\mathcal{S}}$ has cost $1$, by the same argument as that of Theorem~\ref{thMWBCSCwithCost1}, we can set the budget to $2B-1$ and run the polynomial-time $2$-approximation algorithm of Paul et al.~\cite{paul2020budgeted} on this input. This provides us a tree $T=(V(T), E(T))$ with $|E(T)| \le 2B-1$ and $p(T) \ge \frac{1-e^{-1}}{4}opt$. By Lemma~\ref{lmLamprou}, we can decompose $T$ into $4$ subtrees each containing at most $B$ vertices, and then select the subtree maximizing the prize. This concludes the proof.
\end{proof}

\paragraph*{Budgeted Sensor Cover Problem (\textbf{BSCP}).} This problem is very well-known in wireless sensor networks. In \textbf{BSCP}, we are given a set $\mathcal{S}$ of sensors, a set $\mathcal{P}$ of target points in a metric space, a sensing range $R_s$, a communication range $R_c$, and a budget $B$. A target point is covered by a sensor $s$ if it is within a distance $R_s$ of $s$'center. Two sensors are connected if they are at a distance at most $R_c$.
The goal is to find a subset $\mathcal{S}' \subseteq \mathcal{S}$ such that $|\mathcal{S}'| \le B$, the number of covered target points by $\mathcal{S}'$ is maximized and $\mathcal{S}'$ induces a connected subgraph. 

Note that \textbf{BSCP} is a special case of \textbf{MWBCSC}. Then, by Theorem~\ref{thMWBCSCwithCost1}, we have a $2f$-approximation algorithm for \textbf{BSCP}, where $f$ is the maximum number of sensors that cover a target point.

Next we investigate \textbf{BSCP} under the widely used assumption that $R_s\le R_c/2$, where $R_s$ and $R_c$ are the sensing range and communication range of sensors, respectively. In the following, we show that using some materials from the literature leads to an improvement to~\cite{huang2020approximation, kuo2014maximizing, ran2016approximation, xu2021throughput, yu2019connectivity} for \textbf{BSCP} under the assumption that $R_s\le R_c/2$. Note that Huang et al.~\cite{huang2020approximation} studied \textbf{BSCP} under the assumption that $R_s/R_c=O(1)$, however, our result improves their result under the assumption that $R_s\le R_c/2$. 

%over the factors $O(\log{n})$~\cite{gao2018algorithm, ran2016approximation}, $O(\sqrt{B})$~\cite{gao2018algorithm, kuo2014maximizing, xu2021throughput}, $\frac{8(\lceil 2\sqrt{2} C\rceil+1)^2}{1-1/e}$~\cite{huang2020approximation} and  and $\frac{8(\lceil 4C/\sqrt{3}\rceil+1)^2}{1-1/e}$~\cite{yu2019connectivity} where $C=\frac{R_s}{R_c} \le 1$.

\begin{theorem}\label{thBSCPImproved}
If $R_s\le R_c/2$, then \textbf{BSCP} admits a polynomial-time $\frac{8}{1-e^{-1}}$-approximation algorithm.
\end{theorem}

\begin{proof}
We first construct a graph $G=(V, E)$ as follows: $V=\mathcal{S}$ and for any two sensors $s, s' \in \mathcal{S}$, $e=(s, s') \in E$ if $d(s, s') \le R_c$, where $d$ denotes the minimum distance between two senors $s$ and $s'$ in the plane. Now we assign prizes to the nodes in $G$ as follows: we iteratively pick the sensor $s$ that covers the maximum number of the target points in the plane, assign the number of the covered target points to $s$ as its prize and remove $s$ and the covered target points from further consideration. It can be shown that assigning prizes to the nodes in such a way ensures that $G$ contains a connected subgraph $G'$, $G' \subseteq G$, with $|V(G')|\le 2B$ and the prize $p(G') \ge (1-e^{-1})opt$, where $opt$ is the optimum solution to \textbf{BSCP} (see e.g. ~\cite{gao2018algorithm}).%\footnote{Note that in~\cite{gao2018algorithm} each sensor $s$ is endowed with the total satisfaction that it provides to the users, however, this does not effect our proof.}% we have a  Gao et al.~\cite{gao2018algorithm} showed that by the iterative greedy algorithm in which picking the sensor $s$ that covers most of the points in each iteration, assigning the number of covered points to $s$ as its prize and removing $s$ and these points from further consideration, leads to a graph $G'$ with the following properties. $G'$ contains a connected subgraph $G''=(V'', E'')$ such that $|V''| \le 2B$ and $p(G'') \ge (1-e^{-1})opt$, where $opt$ is the optimum.

%By greedy algorithm provided by Gao et al.~\cite{gao2018algorithm}, from the given graph $G$, we first compute a graph $G'$ containing a connected subgraph $G''=(V'', E'')$ such that $|V''| \le 2B$ and $p(V'') \ge (1-e^{-1})opt$.

From the proof of Theorem~\ref{thMWBCSCwithCost1}, we only need to consider the unrooted version of \textbf{E-URAT} where the input graph is $G$ with cost $1$ on the edges and budget $2B-1$. Now we apply the polynomial-time $2$-approximation algorithm of Paul et al.~\cite{paul2020budgeted} to this input and achieve a tree $T$ of cost at most $2B-1$ and prize at least $p(T) \ge \frac{1-e^{-1}}{2}opt$. By Lemma~\ref{lmLamprou}, we can decompose $T$ into $4$ subtrees each containing at most $B$ vertices, and then select the one with the maximum prize. This concludes the proof.
\end{proof}

\end{document}